\documentclass[11pt,a4paper]{article}

\usepackage[a4paper,text={150mm,240mm},centering,headsep=10mm,footskip=15mm]{geometry}
\usepackage[T1]{fontenc}
\usepackage{graphicx}
\usepackage{epsfig, float}
\usepackage{pgf,tikz}
\usetikzlibrary{arrows}

\usepackage{amsmath,amssymb}
\usepackage{amsfonts}
\usepackage{url}
\usepackage{bbm}
\usepackage{mathrsfs} 
\usepackage[english]{babel}
\usepackage{hyperref}

\newcommand{\R}{\mathbb{R}}
\newcommand{\C}{\mathbb{C}}

\newcommand{\id}{\mathbbm{1}}

\newcommand{\diag}{{\rm diag} }
\newcommand{\hastobe}{\stackrel{!}{=}}

\newtheorem{theorem}{Theorem}[section]
\newtheorem{lemma}[theorem]{Lemma}

\newenvironment{proof}[1][Proof:]{\begin{trivlist}
\item[\hskip \labelsep {\bfseries #1}]}{\end{trivlist}}
\newenvironment{proofT}[1][Proof of the theorem:]{\begin{trivlist}
\item[\hskip \labelsep {\bfseries #1}]}{\end{trivlist}}
\newenvironment{definition}[1][Definition:]{\begin{trivlist}
\item[\hskip \labelsep {\bfseries #1}]}{\end{trivlist}}
\newenvironment{example}[1][Example:]{\begin{trivlist}
\item[\hskip \labelsep {\bfseries #1}]}{\end{trivlist}}
\newenvironment{remark}[1][Remark:]{\begin{trivlist}
\item[\hskip \labelsep {\bfseries #1}]}{\end{trivlist}}

\newcommand{\qed}{\hfill\ensuremath{\square}}

\title{A simple explicitly solvable interacting relativistic $N$-particle model}

\author{
Matthias Lienert\thanks{lienert@math.lmu.de} \ and
Lukas Nickel\thanks{nickel@math.lmu.de} \\[0.2cm]
	Mathematisches Institut, Ludwig-Maximilians-Universit\"at\\
	Theresienstr. 39, 80333 M\"unchen, Germany 
}

\date{February 3, 2015}

\begin{document}

\maketitle

\begin{abstract}
\noindent In this paper, we generalize a previous relativistic $1+1$-dimensional model for two mass-less Dirac particles with relativistic contact interactions to the $N$-particle case. Our model is based on the notion of a multi-time wave function which, according to Dirac, is the central object in a relativistic multi-particle quantum theory in the Schrödinger picture. Consequently, we achieve a manifestly Lorentz invariant formulation on configuration space-time. Our model is constructed to be compatible with antisymmetry and probability conservation in a relativistic sense. On the mathematical side, we further develop the method of multi-time characteristics and show that uniqueness of solutions follows from probability conservation. We prove that the model is interacting and outline how one can understand the interaction as effectively given by a $\delta$-potential at equal times. Finally, we answer the question whether Lorentz invariant and probability-conserving 
dynamics can also be obtained when the particles are confined in a region with a non-zero minimal space-like distance, a question relevant for an extension to higher dimensions.
\\

    \noindent \textbf{Keywords:} relativistic interactions, multi-time wave functions, zero-range physics, Dirac equation, boundary conditions
\end{abstract}

\section{Introduction}

In a preceding paper \cite{1d_model}, it was shown that Lorentz transformations of $N$-particle wave functions in the Schrödinger picture lead to the necessity to consider \textit{multi-time wave functions}, i.e.\@ maps 
\begin{equation}
 \psi : \Omega \subset \underbrace{\R^{1+d} \times \cdots \times \R^{1+d}}_{N~{\rm times}} ~\longrightarrow~\mathcal{S},~~~(x_1,..., x_N) ~\longmapsto~\psi (x_1, ... , x_N),
 \label{eq:multitimewavefn}
\end{equation}
where $\mathcal{S}$ is a suitable spin space. The crucial point is that the Lorentz transformation of a simultaneous space-time configuration $((t,\mathbf{x}_1),...,(t,\mathbf{x}_N))$ in general yields $((t_1',\mathbf{x}_1'),...,$ $(t_N',\mathbf{x}_N'))$ where the times $t_k'$ are not equal. This reasoning motivates to regard the set $\mathscr{S}$ of space-like configurations as the appropriate domain $\Omega$. The multi-time wave function $\psi$ generalizes the familiar single-time wave function from Schr\"odinger's theory and gives back the latter for equal times $x_k^0 \equiv t~\forall k$.\\
Following Dirac \cite{dirac_32}, one usually considers $N$ simultaneous first order partial differential equations (PDEs) as evolution equations:
\begin{equation}
 i \frac{\partial}{\partial {x_k^0}} \psi ~=~ H_k \psi,~~k=1,...,N.
 \label{eq:multitimeevolution}
\end{equation}
Schr\"odinger's equation is re-obtained from \eqref{eq:multitimeevolution} by the chain rule for $\psi(t,\mathbf{x}_1,...,t,\mathbf{x}_N)$, with Hamiltonian $H = \sum_{k=1}^N H_k$.\\
The strategy employed in \cite{1d_model} to escape no-go theorems about relativistic interactions such as for interaction potentials \cite{nogo_potentials} was to prescribe the free multi-time Dirac equation on $\mathscr{S}$ and to introduce interaction only by boundary conditions on the set $\mathscr{C}$ of coincidence points in space-time. This idea is related to the field of zero-range physics \cite{albeverio}. It was worked out in detail for two mass-less particles in $1+1$ dimensions ($d=1$). In this case, an explicit solution of the model was feasible via a generalized version of the method of characteristics. This method made it possible to prove existence and uniqueness of classical solutions. Probability conservation on general space-like hypersurfaces, however, had to be checked separately. It led to certain conditions on the conserved tensor current of the theory which were equivalent to linear relations between the components of the wave function. Together with the requirement of Lorentz invariance, this 
allowed to formulate a 
general class of boundary conditions. This class was also compatible 
with antisymmetry of the wave function in the case of indistinguishable particles and led to interaction in the sense that a generic initial product wave function becomes entangled with the time evolution.\\
In this paper, we extend the previous model to the $N$-particle case. In order to achieve a concise formulation, we use a bottom-up approach: instead of starting from all mathematically possible classes of boundary conditions which lead to existence and uniqueness of solutions and successively restricting the class of boundary conditions according to the various physical requirements, we directly specialize on the case of indistinguishable particles, extract a class of physically reasonable boundary conditions and then prove existence and uniqueness for the resulting model.\\
The structure of the paper is as follows. In section \ref{sec:defmodel} we introduce the model as defined by its multi-time equations, domain, initial values and boundary conditions. In sec.\@ \ref{sec:simplifications}, the implications of antisymmetry are discussed and the general solution of the multi-time equations is found using a generalized version of the method of characteristics. Sec.\@ \ref{sec:SectionProbConserv} deals with the question how to formulate probability conservation on space-like hypersurfaces. A class of boundary conditions is extracted which ensures it. Furthermore, we show that the uniqueness of solutions of the multi-time equations follows from probability conservation. In sec.\@ \ref{sec:li} the requirements of Lorentz invariance are checked and the previous class of boundary conditions is shown to satisfy them. The main result of the paper is an existence and uniqueness theorem (sec.\@ \ref{sec:SectionExistenceandUniquenessSpacelikeConf}). In sec.\@ \ref{sec:interaction} we give 
a general argument that the model is interacting and show that one can regard the interaction as given by an effective $\delta$-potential at equal times. Moreover, in sec. \ref{sec:spacelikeconfigs} we answer a different question that was raised in \cite{1d_model}: Do consistent Lorentz invariant and probability-conserving dynamics exist on the set $\Omega_\alpha$ of space-like configurations with a minimum space-like distance $\alpha$? This question is relevant for the question whether one can also achieve interaction effects by boundary conditions in higher space-time dimensions.

\section{Definition of the model} \label{sec:defmodel}

Our model is based on a multi-time wave function \eqref{eq:multitimewavefn} for $N$ mass-less Dirac particles on the set of space-like configurations,
\begin{equation}
 \mathscr{S} ~:=~ \{ (t_1,z_1,...,t_N,z_N) \in \R^{2N} : (t_j-t_k)^2 - (z_j-z_k)^2 < 0~\forall j \neq k\},
 \label{eq:spacelikeconfigs}
\end{equation}
in $1+1$-dimensional space-time, with metric $g = \diag(1,-1)$. The appropriate spin space is $\mathcal{S} = (\C^2)^{\otimes N}$. Thus, $\psi$ has $2^N$ spin components $\psi_i,~i = 1,...,2^N$.\\
As multi-time evolution equations \eqref{eq:multitimeevolution} we use a system of $N$ mass-less Dirac equations
\begin{equation}
\label{eq:multitimedirac}
i  \gamma_k^{\mu} \partial_{k, \mu} \, \psi (x_1, ..., x_N) ~=~ 0 , ~~k= 1,  ..., N.
\end{equation}
Here, $x_k = (t_k,z_k)$, $\partial_{k, \mu} = \frac{\partial}{\partial x_k^{\mu}}$ and $\gamma_k^\mu$ is the $\mu$-th Dirac gamma matrix acting on the spin index of the $k$-th particle. We choose the following representation:
\begin{equation}
 \gamma^0 = \sigma_1 = \left( \begin{array}{cc} 
0 & 1 \\ 1 & 0
\end{array} \right), ~~~ \gamma^1 = \sigma_1 \sigma_3 = \left( \begin{array}{cc} 
0 & -1 \\ 1 & 0
\end{array} \right),
\label{eq:gammamatrices}
\end{equation}
where $\sigma_i,~i = 1,2,3$ denote the Pauli matrices. This representation diagonalizes \eqref{eq:multitimedirac} which can be seen by multiplying eq.\@ \eqref{eq:multitimedirac} with $\gamma^0_k$ from the left. This results in
\begin{equation}
  \left( \frac{\partial}{\partial t_k} + \sigma_{3, k} \frac{\partial}{\partial z_k} \right) \psi (t_1, z_1, ..., t_N, z_N) ~=~ 0,~~k= 1,  ..., N,
 \label{eq:multitimedirac2}
\end{equation}
where $\sigma_{3, k}$ is the third Pauli matrix, $\sigma_3 = \diag(1,-1)$, acting on the spin index of the $k$-th particle.\\
Initial data are prescribed on the set
\begin{equation}
 \mathcal{I}~:=~ \{ (t_1,z_1,...,t_N,z_N) \in \overline{\Omega} : t_1 = \dots = t_N = 0 \}.
 \label{eq:initialdata}
\end{equation}
Since $ \mathscr{S}$ has a non-empty boundary  $\partial  \mathscr{S}$, one should expect that boundary conditions are needed to ensure the uniqueness of solutions. At this point, we leave open the exact nature of the boundary conditions. It will be clarified by further considerations about Lorentz invariance and probability conservation.\\
We summarize the structure of the model as follows:
\begin{equation}
 \left\lbrace \begin{array}{l} \mathrm{the \ system \ of \ equations \ \eqref{eq:multitimedirac2} \ on \ } \Omega = \mathscr{S},
\\ \mathrm{initial \ conditions \ on \ } \mathcal{I},
\\ \mathrm{boundary \ conditions \ on \ } \partial \Omega.
\end{array}  \right.
\label{eq:model}
\end{equation}

\section{Antisymmetry, general solution and multi-time characteristics} \label{sec:simplifications}

In this section, we first show how antisymmetry of the wave function for indistinguishable particles allows to reduce the domain from $N$ disconnected parts to a single connected one. Using a notation for the spin components which is tailor-made for the multi-time equations \eqref{eq:multitimedirac2} we show how this facilitates to explicitly determine their general solution. This leads to the notion of multi-time characteristics.

\subsection{Antisymmetry and reduction of the domain}
Following the spirit outlined in the introduction, we make simplifications wherever physically reasonable in order to achieve a concise model for which existence and uniqueness can be proved elegantly. The first simplification is the assumption of indistinguishable particles. This is natural considering that the particles are not dynamically distinguished by eqs.\@ \eqref{eq:multitimedirac} alone.
 Denote the spin components of $\psi$ by $\psi_{s_1 ... s_N}$ where each $s_i$ can take the values $\pm 1$. We write
\begin{equation}
 \left( \begin{array}{c}
\psi_1 \\ \psi_2 \\ \psi_3 \\ \vdots \\ \psi_{2^N} 
\end{array} \right) ~\equiv~ \left( \begin{array}{c}
\psi_{--\dots--} \\ \psi_{--\dots -+} \\ \psi_{--\dots +-} \\ \vdots \\ \psi_{++\dots++} 
\end{array} \right).
 \label{eq:spincpts}
\end{equation}
Indistinguishability implies the following antisymmetry condition for the wave function. Let $\pi \in S^N$ be a permutation. Then
\begin{equation}
 \psi_{s_{\pi(1)} ... s_{\pi(N)}}(x_{\pi(1)},...,x_{\pi(N)})~\stackrel{!}{=}~ (-1)^{\mathrm{sgn}(\pi)} \psi_{s_1 ... s_N}(x_1,...,x_N).
 \label{eq:antisymmetrycond}
\end{equation}
We now use this condition to relate a solution of \eqref{eq:model} on the different parts of the domain $\Omega = \mathscr{S}$ (see eq.\@ \eqref{eq:spacelikeconfigs}). Note that in one spatial dimension, $\mathscr{S}$ separates into $N!$ disjoint parts which can be classified according to the relation of the spatial coordinates $z_k$, e.g. $z_2 < z_1 < z_5 < z_3 < \dots $. Using the permutation group $S^N$, we write $\mathscr{S}$ as the disjoint union of open sets as follows:
\begin{align}
  \mathscr{S}~ &\,=~\bigsqcup_{\pi \in S^N} \mathscr{S}_{\pi},\nonumber\\
\mathrm{where}~~~\mathscr{S}_{\pi} ~&:=~ \left\{(t_1, z_1, ..., t_N, z_N) \in \mathscr{S} : z_{\pi(1)} < \dots < z_{\pi(N)} \right\}.
\label{eq:omegapi}
\end{align}
The crucial point is the following: given a solution of the model, as defined by \eqref{eq:model} on $\mathscr{S}_1$ (corresponding to $z_1 < \dots < z_N$), antisymmetric continuation via eq.\@ \eqref{eq:antisymmetrycond} yields a solution on $\mathscr{S}_\pi$, provided the boundary and initial conditions are chosen to be compatible with antisymmetry. Note that this reduces the possible classes of initial boundary value problems (IBVPs) \eqref{eq:model} to an autonomous IBVP on $\mathscr{S}_1$. We shall employ this strategy in the following. Our new model may be summarized according to \eqref{eq:model} with $\mathscr{S}$ replaced by $\mathscr{S}_1$.

\subsection{Multi-time characteristics and general solution}
Using the notation \eqref{eq:spincpts}, we express the diagonalized multi-time Dirac equations \eqref{eq:multitimedirac2} for a fixed component $\psi_{s_1 ... s_N}$ as follows:
\begin{equation}
 \left( \frac{\partial}{\partial t_k} -  s_k \frac{\partial}{\partial z_k} \right) \psi_{s_1 ... s_k... s_N} ~=~ 0,~~k = 1, ..., N.
 \label{eq:multitimediracspincpts}
\end{equation}
Note that \eqref{eq:multitimediracspincpts} imposes $N$ equations for each of the $2^N$ spin components $\psi_{s_1 ... s_N}$. This simple form of the equations allows to find the general solution.

\begin{lemma}
 The general solution of eqs.\@ \eqref{eq:multitimediracspincpts} is given by
 \begin{equation}
  \psi_{s_1...s_N}(t_1,z_1,...,t_N,z_N) ~=~ f_{s_1...s_N} \left( z_1 +s_1 t_1, ..., z_N +s_N t_N \right)
  \label{eq:generalsolution}
 \end{equation}
 where $f_{s_1...s_N} \hspace{-0.1cm} : \R^N \rightarrow \C$ are $C^1$-functions, $s_1 = \pm 1, ..., s_N = \pm 1$.
 \label{thm:generalsolution}
\end{lemma}

\begin{proof}
 The result is obvious if one is familiar with the notation. Simply write out eq.\@ \eqref{eq:multitimediracspincpts} for $\psi_{s_1 ... s_N}$ separately: $\left( \tfrac{\partial}{\partial t_1} -  s_1 \tfrac{\partial}{\partial z_1} \right)\psi_{s_1 ... s_N} = 0$, ..., $\left( \tfrac{\partial}{\partial t_N} -  s_N \tfrac{\partial}{\partial z_N} \right)\psi_{s_1 ... s_N} = 0$. This implies the form \eqref{eq:generalsolution}. \qed
\end{proof}

The form of the general solution motivates the following definition.

\begin{definition}
Let $ p = (t_1, z_1, \dots t_N, z_N) \in \mathbb{R}^{2N}$. Then we call
\begin{equation}
 c_k ~:=~  z_k +s_k t_k
\label{eq:characteristicvalues}
\end{equation}
the \textit{characteristic values} at $p$ associated with the component $\psi_{s_1 ... s_N}$.\\
Furthermore, we define the \textit{multi-time characteristics} of the components $\psi_{s_1 ... s_N}$ by
\begin{equation}
 S_{s_1 ... s_N} (c_1, ..., c_N) ~:=~ \{ (t_1,z_1,...,t_N,z_N) \in \R^{2N} : z_k + s_k t_k = c_k\}.
 \label{eq:multitimecharacteristic}
\end{equation}
\end{definition}

With these definitions, one can reformulate lemma \ref{thm:generalsolution} as follows: the components $\psi_{s_1 ... s_N}$ of solutions of \eqref{eq:multitimediracspincpts} are constant on the respective multi-time characteristics \eqref{eq:multitimecharacteristic}. Note that this implies existence and uniqueness on the domain $\mathbb{R}^{2N}$ for an initial value problem at $t_1 = ... = t_N = 0$, the functions $f_{s_1...s_N}$ being given by the initial values. However, as known from \cite{1d_model}, this is in general not true for a domain with boundary such as $\mathscr{S}_1$.

%
%
%
%

\section{Probability conservation} \label{sec:SectionProbConserv}

In this section, we specify an adequate notion of probability conservation on space-like hypersurfaces using the conserved tensor current of the multi-time Dirac equations (sec.\@ \ref{sec:relprobcons}). This enables us to give a sufficient condition on the components of the wave-function which leads to probability conservation (sec.\@ \ref{sec:probconsbdc}). Furthermore, we prove a general theorem showing that probability conservation implies the uniqueness of solutions (sec.\@ \ref{sec:probconsunique}).

\subsection{A geometric formulation of probability conservation} \label{sec:relprobcons}

It is clear that the non-relativistic notion of probability conservation, $\int d^d x_1 \cdots d^d x_N $ $|\psi|^2 (t,\mathbf{x}_1,...,t,\mathbf{x}_N) = 1 \ \forall t$, which heavily draws on a notion of simultaneity, has to be generalized in a relativistic context. Building on previous work, such a generalization was given and justified in \cite[sec.\@ 4]{1d_model}. It makes use of the conserved tensor current of the multi-time Dirac equations \eqref{eq:multitimedirac}, i.e.\@
\begin{equation}
j^{\mu_1 ... \mu_N} ~:=~ \overline{\psi} \gamma^{\mu_1}_1 \dots \gamma^{\mu_N}_N \psi~~~\mathrm{with}~~~\partial_{k,\mu_k} j^{\mu_1...\mu_k...\mu_N} ~=~ 0~\forall k.
\label{eq:j}
\end{equation}
 We define the \textit{current form} as follows:
\begin{flalign}
\omega_j ~ := ~\sum_{\mu_1,...,\mu_N = 0}^{d} (-1)^{\mu_1 + \dots + \mu_N} j^{\mu_1 ... \mu_N} & (dx^0_1 \wedge \dots \widehat{dx^{\mu_1}_1} \wedge \dots \wedge dx^d_1)\nonumber \\
  \wedge\cdots \wedge & (dx^0_N \wedge \dots \widehat{dx^{\mu_N}_N} \dots \wedge dx^d_N)
\label{eq:omegaj}
\end{flalign}
where $\widehat{(\cdot)}$ denotes omission from the wedge product. The continuity equations \eqref{eq:j} imply that the exterior derivative of $\omega_j$ vanishes, i.e.\@ $d \omega_j=0$. The relativistic notion of probability conservation on a domain $\Omega$ then reads \cite{1d_model}:
\begin{equation}
 \int_{\Sigma^N \cap \, \Omega}  \omega_j ~=~  \int_{(\Sigma')^N \cap \,  \Omega}  \omega_j
 \label{eq:probcons}
\end{equation}
for all pairs of space-like hypersurfaces\footnote{We assume all space-like hypersurfaces to be smooth for the rest of the paper.} $\Sigma,\Sigma'$.

\subsection{Boundary conditions from probability conservation} \label{sec:probconsbdc}

The formulation via the $Nd$-form $\omega_j$ together with the property $d \omega_j = 0$ makes it possible to use Stokes' theorem to extract conditions on $\omega_j$ and thereby on $j$ which ensure probability conservation.

\begin{lemma}
Probability conservation on $\mathscr{S}_1$ in the sense of
\begin{equation} \label{eq:probconservconditiononOmega1}
 \int_{\Sigma^N \cap \, \mathscr{S}_1}  \omega_j ~=~  \int_{(\Sigma')^N \cap \,  \mathscr{S}_1}  \omega_j
\end{equation}
for all space-like hypersurfaces $\Sigma, \Sigma'$ holds if the wave function $\psi$ is compactly supported on all sets of the form $\Sigma^N \cap \mathscr{S}_1$ and if
 \begin{equation}\label{eq:omegajiszero}
\left. \omega_j \right|_{\mathscr{C}_1}~ =~ 0,
\end{equation}
 where
\begin{equation}
\mathscr{C}_1 :=  \left\lbrace (t_1, z_1, \dots t_N, z_N) \in \partial \mathscr{S}_1 \,|
\ \exists k: t_k = t_{k+1} \wedge z_k = z_{k+1} \right\rbrace.
\label{eq:c1}
\end{equation}
\label{thm:probcons}
\end{lemma}

\begin{remark}
\begin{enumerate}
 \item The assumption of compact support of the wave function (or, alternatively, of suitable drop-off conditions) with respect to spatial directions is needed as a technical assumption in the proof. It is reasonable because the multi-time Dirac equations have finite propagation speed (see eq.\@ \eqref{eq:generalsolution}). Consequently, compactly supported initial data imply the desired property.
 \item Note that the wave function is, strictly speaking, not defined on $\partial \mathscr{S}_1$. When using values of the wave function at the boundary (such as in eq.\@ \eqref{eq:omegajiszero}), we assume that the wave function is continuous\footnote{The assumption of continuity is justified in sec.\@ \ref{sec:SectionExistenceandUniquenessSpacelikeConf} where it is shown that a unique $C^k$ solution exists for an appropriate IBVP.} and refer to the corresponding limit in $\mathscr{S}_1$. In this way, jumps of the wave function across the boundaries of different $\mathscr{S}_\pi$ are admitted. In fact, singularities of this kind are typical for zero-range interactions \cite[appendix J]{albeverio}.
\end{enumerate}

\end{remark}

\begin{proof}
We adopt the idea of \cite[proof of thm.\@ 4.4]{1d_model} and generalize it for $N$ particles. Let $\Sigma,\Sigma'$ be space-like hypersurfaces. We construct a suitable submanifold with boundary in order to be able to use Stokes' theorem.\\
Let $t_{\Sigma} (z)$ denote the time coordinate of the unique point $p = (t_{\Sigma} (z),z) \in \Sigma$.
Let $R > 0$ and consider the following set:
\begin{equation} \label{eq:StokesVolumen}
V_R ~:=~ \left\lbrace (t_1, z_1, \dots, t_N, z_N) \in \overline{\mathscr{S}}_1 \left| \begin{array}{c} \exists \tau \in \left[ 0, 1 \right]: \forall k: t_k = t_{\Sigma} (z_k) + \tau \left( t_{\Sigma'} (z_k) - t_{\Sigma} (z_k) \right)\\ \mathrm{and} \ |z_k| \leq R \end{array} \right. \right\rbrace
\end{equation}
$V_R$ is a bounded and closed, thus compact, $(N+1)$-dimensional submanifold of $\R^{2N}$ with boundary
\begin{equation}
\partial V_R ~=~ (\Sigma^N \cap \mathscr{S}_1) \cup ((\Sigma')^N \cap \mathscr{S}_1) \cup M_1 \cup M_2
\end{equation}
where $M_2$ is the subset of $V_R$ with $|z_k| = R$ for some $k$ and
\begin{equation}
M_1 ~=~ V_R \cap \partial \mathscr{S}_1.
\end{equation} 
Because of the first condition in the definition of $V_R$, a configuration in $V_R$ is always an element of $\mathcal{S}^N$ for some space-like hypersurface $\mathcal{S}$. Therefore, it can only be an element of $M_1 \subset \partial \mathscr{S}_1$ (i.e.\@\texttt{} light-like) if $\exists k: t_k = t_{k+1}$ and $z_k = z_{k+1}$. This implies $M_1 \subset \mathscr{C}_1$.\\
In the limit $R \rightarrow \infty$, the integral $\int_{M_2} \omega_j$ vanishes because of the compact support of the wave function. Thus, it follows from the the theorem of Stokes, together with $d \omega_j = 0$, that
\begin{equation}
0 ~=~ \lim_{R \rightarrow \infty} \int_{V_R} d \omega_j = \lim_{R \rightarrow \infty} \int_{\partial V_R} \omega_j =  - \int_{\Sigma^N \cap \mathscr{S}_1} \omega_j + \int_{(\Sigma')^N \cap \mathscr{S}_1} \omega_j + \int_{M_1} \omega_j  .
\end{equation}
The minus sign in front of the first integral on the r.h.s.\@ is due to orientation conventions. Thus, probability conservation in the sense of eq.\@ \eqref{eq:probconservconditiononOmega1} holds iff $\int_{M_1} \omega_j = 0$. In order to make this integral vanish for all possible choices of $\Sigma, \Sigma'$, the condition 
\begin{equation}
\left. \omega_j \right|_{\mathscr{C}_1} ~\stackrel{!}{=}~ 0
 \label{eq:omegajcondproof}
\end{equation}
has to be satisfied. \qed
\end{proof}

Next, we study the implications of condition \eqref{eq:omegajcondproof} for the components of the wave function. For simplicity, we first focus on the special case of equal-time hypersurfaces in a fixed but otherwise arbitrary Lorentz frame.

\begin{lemma} \label{thm:karltheodor}
Let $\mathscr{C}_{1,t} := \{ (t_1,z_1,...,t_N,z_N) \in \mathscr{C}_1 :  t_1 = \dots = t_N\}$. Then the condition for probability conservation on equal-time hypersurfaces $\Sigma_t$ in a particular Lorentz frame, i.e.\@ \eqref{eq:omegajcondproof} with $\mathscr{C}_1$ replaced by $\mathscr{C}_{1,t}$, holds if and only if the following condition is satisfied:
\begin{equation}
\label{eq:EquationforPsiThatOmegaVanishes}
\psi^{\dagger}(p) \left( \sigma_{3,k} - \sigma_{3, k+1} \right) \psi(p) ~=~ 0 ~~~\forall p \in \mathscr{C}_{1,t}^{(k)}~~\forall k=1, ..., N-1,
\end{equation}
where $\mathscr{C}_{1,t}^{(k)} := \left\lbrace (t_1,z_1,...,t_N,z_N) \in \mathscr{C}_{1,t} : z_k = z_{k+1}\right\rbrace$.\\
Furthermore, eq.\@ \eqref{eq:EquationforPsiThatOmegaVanishes} can be rewritten as
\begin{equation} \label{eq:EquationforPsiinComponentsthatOmegavanishes}
\sum_{\substack{(s_1,\dots,s_N) \in \left\lbrace -, + \right\rbrace^N \\ s_k \neq s_{k+1}}}
s_{k+1} |\psi_{s_1 \dots s_N}|^2(p) ~=~ 0~~\forall p \in \mathscr{C}_{1,t}^{(k)}.
\end{equation}
\end{lemma}

\begin{proof}
We have to evaluate the condition $\left. \omega_j \right|_{\mathscr{C}_{1,t}} = 0$. Note that for $p \in \mathscr{C}_{1,t}$ there exists a $k \in \left\lbrace 1, ... , N-1 \right\rbrace$ such that $p = (t, z_1, ..., t, z_k = z, t, z_{k+1} = z, ..., t, z_N)$.\\
Next, we calculate $\left. \omega_j \right|_{\mathscr{C}_{1,t}}$ according to eq.\@ \eqref{eq:omegaj}, recalling that in this case $x_k^0 = t$ and $x_j^1 = z_j$ as well as $z_k = z_{k+1} = z$.  All terms with more than one index $\mu_l = 1$ in $j^{\mu_1...\mu_l...\mu_N}$ vanish because they contain $dt \wedge dt = 0$. Moreover, the terms with $\mu_k = \mu_{k+1}=0$ do not contribute, either, as they contain $dz \wedge dz = 0$. We are left with terms where all indices $\mu_j$ are equal to zero apart from the $k$-th or the $(k+1)$-th: 
\begin{align}
\omega_j(p) ~&=~ - j^{0 ...(\mu_k = 0) (\mu_{k+1} = 1)... 0} (p) \, dz_1 \wedge \dots \wedge dz_{k-1} \wedge dz \wedge dt \wedge dz_{k+2} \wedge \dots \wedge dz_N \nonumber 
\\ &~~~ -j^{0 \dots (\mu_k = 1) (\mu_{k+1} =  0) \dots 0} (p) \, dz_1 \wedge  \dots \wedge dz_{k-1} \wedge dt \wedge dz \wedge d z_{k+2} \wedge \dots \wedge dz_N \nonumber
\\ &=~ \left( j^{0 \dots 1 0 \dots 0} - j^{0 \dots 0 1 \dots 0} \right)(p)\, dz_1 \wedge \dots \wedge dz_{k-1} \wedge dz \wedge dt \wedge dz_{k+2}\wedge \dots \wedge dz_N 
\end{align}  
This expression vanishes if and only if the bracket is zero. This yields condition \eqref{eq:EquationforPsiThatOmegaVanishes}:
\begin{equation}
 0 ~=~ (j^{0 \dots 1 0 \dots 0} - j^{0 \dots 0 1 \dots 0})(p) ~=~ \psi^{\dagger}(p) \left( \sigma_{3,k} - \sigma_{3, k+1} \right) \psi(p)
 \label{eq:jcond1}
\end{equation}
Written out in components, eq.\@ \eqref{eq:jcond1} reads:
\begin{equation}
 0 ~=~\sum_{(s_1,\dots,s_N) \in \left\lbrace \pm 1 \right\rbrace^N}
  \left( - s_k |\psi_{s_1 \dots s_N}|^2(p) +  s_{k+1} |\psi_{s_1 \dots s_N}|^2(p) \right),
\end{equation}
where $k$ was defined above.\\
Summands with $s_k = s_{k+1}$ cancel out. We are left with
\begin{equation}
0 ~=~  \sum_{\substack{(s_1,\dots,s_N) \in \left\lbrace -, + \right\rbrace^N \\ s_k \neq s_{k+1}}}
 2 s_{k+1} |\psi_{s_1 \dots s_N}|^2(p).
\end{equation}
which yields \eqref{eq:EquationforPsiinComponentsthatOmegavanishes}. \qed
\end{proof}

We now take the following approach in order to find adequate boundary conditions that lead to probability conservation on general space-like hypersurfaces. First, we choose a subclass of \eqref{eq:EquationforPsiinComponentsthatOmegavanishes} which turn out to be Lorentz invariant (see the next section) and to ensure the existence of a solution (see sec.\@ \ref{sec:SectionExistenceandUniquenessSpacelikeConf}). Then we prescribe the condition on the whole set $\mathscr{C}_1$ and show that it is indeed sufficient to ensure condition \eqref{eq:omegajiszero} and therefore probability conservation on general hypersurfaces.\\
It is useful to define the sets
\begin{equation}
\mathscr{C}_{k, k+1} ~:=~ \left\lbrace (t_1, z_1, ..., t_N, z_N) \in \overline{\mathscr{S}}_1 |\, \exists k: 
 t_k = t_{k+1} \wedge z_k = z_{k+1} \right\rbrace.
 \label{eq:ckkplusone}
\end{equation}
One can then write $\mathscr{C}_1 = \bigcup_{k = 1}^{N-1} \mathscr{C}_{k,k+1}$ (see eq.\@ \eqref{eq:c1}).

\begin{lemma} \label{thm:MyBCareProbconserving}
Let  $\varphi^{(k)} \in (-\pi,\pi]$ for $k=1,..., N-1$. Then the boundary conditions
\begin{equation} \label{eq:probconsbdc}
\psi_{s_1 ... s_{k-1} + - s_{k+2}... s_N} ~\hastobe~  e^{i \varphi^{(k)}} \psi_{s_1 ... s_{k-1} - + s_{k+2}... s_N} ~~ \mathrm{on} ~~ \mathscr{C}_{k, k+1},~~k=1,...,N-1
\end{equation}
imply probability conservation on all space-like hypersurfaces in the sense of eq.\@ \eqref{eq:probconservconditiononOmega1}.
\end{lemma}

\begin{proof}
It was shown in lemma \ref{thm:probcons} that equation \eqref{eq:probconservconditiononOmega1} follows if 
\begin{equation}
\left. \omega_j \right|_{\mathscr{C}_1} ~=~ 0 . 
\end{equation}
We show that this equation indeed holds.
Pick a point $p \in \mathscr{C}_1$. Then $\exists k: p \in \mathscr{C}_{k, k+1}$. Condition \eqref{eq:probconsbdc} at this point yields:
\begin{equation}
 |\psi_{s_1... s_{k-1} + -s_{k+2} ...s_N}|^2 (p) ~=~ |\psi_{s_1 ... s_{k-1} - + s_{k+2}...s_N}|^2 (p).
\end{equation}
It follows that 
\begin{equation}
j^{\mu_1 \mu_2 \dots 0 1 \dots \mu_N} (p) ~=~ j^{\mu_1 \mu_2 \dots 1 0 \dots \mu_N} (p)  
\label{eq:jantisymmetriconC}
\end{equation}
because the expression for the current is diagonal in the components. In the formula for $\omega_j(p)$ (eq.\@ \eqref{eq:omegaj}), we can first sum over the indices $\mu_k, \mu_{k+1}$ and afterwards over the rest. Then there are four possibilities in the summands: 
\begin{itemize}
\item $(\mu_k, \mu_{k+1}) = (0,0)$ or $(1,1)$: These do not contribute because the coordinates of the $k$-th and $(k+1)$-th particles are equal, say to $(t,z)$, so either $dz \wedge dz = 0$ or $dt \wedge dt = 0$ appears as a factor in the wedge product.
\item $(\mu_k, \mu_{k+1}) = (0,1)$ or $(1,0)$. One can see that these two factors cancel each other because (abbreviating the other factors in the wedge product by $A$ and $B$)
\begin{flalign} \nonumber & j^{\mu_1 ... (\mu_k =0)(\mu_{k+1}= 1)... \mu_N} A \wedge dt\wedge dz \wedge B  + j^{\mu_1 ...(\mu_k =1)(\mu_{k+1}= 0)... \mu_N} A \wedge dz \wedge dt \wedge B
\\ = & ~\left( j^{\mu_1 \dots 0 1 \dots \mu_N} (p) - j^{\mu_1 \dots 1 0 \dots \mu_N} (p) \right) A \wedge dt \wedge dz \wedge B~ \overset{\eqref{eq:jantisymmetriconC}}{=} ~0 .
\end{flalign}
\end{itemize}
Therefore, the probability-conserving property $\omega_j (p) = 0$ holds.  \qed
\end{proof}

\subsection{Probability conservation implies uniqueness of solutions} \label{sec:probconsunique}
 The notion \eqref{eq:probcons} of probability conservation is very powerful. In this section, we work out the claim in \cite{1d_model} that $\int_{\Sigma^N \cap \Omega} \omega_j$ is a so-called \textit{energy integral} and that therefore probability conservation implies uniqueness of solutions in a suitable sense. 

\begin{definition}
 Let $\Sigma$ be a space-like hypersurface. We define function spaces
\begin{equation}
 \mathcal{H}_\Sigma^{(N)} ~:=~ L^2(\Sigma^N\cap \Omega) \otimes (\C^2)^{\otimes N}.
 \label{eq:hn}
\end{equation}
Furthermore, we call the solution of the IBVP \eqref{eq:model} \textit{weakly unique} iff for every two solutions $\psi, \varphi$ and every space-like hypersurface $\Sigma$ the restrictions $\psi_{|_{\Sigma}}$, $\varphi_{|_{\Sigma}}$ of $\psi, \varphi$ to arguments in $\Sigma^N \cap \Omega$ are equal as elements of $\mathcal{H}_\Sigma^{(N)}$.
 \label{def:hn}
\end{definition}

\begin{theorem} \label{thm:ProbgivesUniquenessTheorem}
Consider the IBVP \eqref{eq:model} with boundary conditions ensuring probability conservation \eqref{eq:probcons} and initial values on $\mathcal{I} = (\Sigma_0)^N \cap \Omega$, i.e.\@ $\psi_{|_\mathcal{I}} \equiv g \in \mathcal{H}_{\Sigma_0}^{(N)}$. Then its solution is weakly unique.
\end{theorem}

\begin{proof}
Consider the expression
\begin{equation}
 \| \phi \|_\Sigma^2 ~:=~ \int_{\Sigma^N \cap \Omega} \omega_j(\phi),
 \label{eq:norm}
\end{equation}
where $\omega_j(\phi)$ is the current form constructed from $\phi$ according to eqs.\@ \eqref{eq:j} and \eqref{eq:omegaj}.  Because the Dirac tensor current $j$ is positive-definite and sesquilinear in the wave function, $\|\cdot\|_\Sigma$ defines a norm on $\mathcal{H}_\Sigma^{(N)}$.\\
Let $\psi, \varphi$ be solutions of the IBVP. Then: $\psi_{|_{\Sigma_0}} \equiv \varphi_{|_{\Sigma_0}} \equiv g \in \mathcal{H}_{\Sigma_0}^{(N)}$ and therefore $\| \psi_{|_{\Sigma_0}} - \varphi_{|_{\Sigma_0}} \|_{\Sigma_0} = 0$. Now let $\Sigma$ be an arbitrary space-like hypersurface. Probability conservation \eqref{eq:probcons} yields:
\begin{equation}
\| \psi_{|_{\Sigma}} - \varphi_{|_{\Sigma}} \|_\Sigma  ~=~ \| \psi_{|_{\Sigma_0}} - \varphi_{|_{\Sigma_0}} \|_{\Sigma_0} ~=~ 0
\end{equation}
and it follows that $\psi_{|_{\Sigma}} \equiv \varphi_{|_{\Sigma}}$ as elements of $\mathcal{H}_\Sigma^{(N)}$.  \qed
\end{proof}

\begin{remark}
 The proof of thm.\@ \ref{thm:ProbgivesUniquenessTheorem} suggests that the map
 \begin{equation}
  U_{\Sigma \rightarrow \Sigma'}: \mathcal{H}_\Sigma^{(N)} \rightarrow \mathcal{H}_{\Sigma'}^{(N)},~~~\psi_{|_{\Sigma}} \mapsto \psi_{|_{\Sigma'}},
  \label{eq:unitaryevolution}
 \end{equation}
 which sends the restriction of a solution $\psi$ of the IBVP to $\Sigma^N \cap \Omega$ to its restriction to $(\Sigma')^N \cap \Omega$, defines a unitary evolution map from one space-like hypersurface to another (see also \cite[sec. 3]{hbd_subsystems}). The analogous view in quantum field theory constitutes the Tomonaga-Schwinger picture \cite{tomonaga,schwinger}.\\
 Note that having in mind a functional-analytic view on time evolution, it might seem natural to take the reverse way to define a multi-time evolution, i.e.\@ first defining the spaces $\mathcal{H}_\Sigma^{(N)}$ and a unitary map $U_{\Sigma \rightarrow \Sigma'}$. However, this is not convincing because then there may exist $\Sigma \neq \Sigma'$ with $\Sigma \cap \Sigma' \neq \emptyset$ such that $\psi_{|_{\Sigma}}(q) \neq \psi_{|_{\Sigma'}}(q)$ even for $q \in \Sigma \cap \Sigma'$ (see \cite{qftmultitime}). Without additional conditions to enforce $\psi_{|_{\Sigma}}(q) = \psi_{|_{\Sigma'}}(q)$ for $q \in \Sigma \cap \Sigma'$, this would mean that one could not regard the multi-time wave functions and the tensor current $j$ as geometrical objects. This, however, may be necessary for a consistent physical interpretation (see e.g.\@ \cite{hbd}).
\end{remark}

\section{Lorentz invariance} \label{sec:li}

The Lorentz invariance of our model requires the invariance of the domain, the multi-time wave equations and the boundary conditions. Apart from the last point, the invariance is already manifest. In this section, we show that also the class \eqref{eq:probconsbdc} of probability-conserving boundary conditions is Lorentz invariant, meaning that the Lorentz-transformed boundary conditions are satisfied as a consequence of the old ones.
\\The transformation behaviour of spinors under a Lorentz transformation $\Lambda: x \mapsto x'$ in the proper Lorentz group $\mathcal{L}_+^\uparrow$ is given by
\begin{equation}
\psi' (x_1,...,x_N) ~=~ S(\Lambda) \otimes \dots \otimes S(\Lambda)  \psi (\Lambda^{-1}x_1,...,\Lambda^{-1}x_N) 
\label{eq:spinortrafo}
\end{equation}
where 
\begin{equation}
 S(\Lambda) ~=~ \exp \left( - \frac{i}{4} \omega_{\mu \nu} \sigma^{\mu \nu} \right) ,~~~ \sigma^{\mu \nu} = \frac{i}{2} \left[ \gamma^{\mu}, \gamma^{\nu} \right].
\end{equation}
Here, $\omega$ is an antisymmetric $(1+d)\times(1+d)$ matrix which characterizes $\Lambda$.\\
For $d = 1$, there is only one free parameter $\beta \in \R$ corresponding to a boost in $z$-direction. One obtains
\begin{equation}
S(\Lambda) ~=~ \left( \begin{array}{cc} 
\cosh \beta + \sinh \beta & 0 \\ 0 & \cosh \beta - \sinh \beta
\end{array} \right).
\end{equation}
As the matrix is diagonal due to our choice of $\gamma$-matrices, it is easy to calculate the $N$-fold tensor product in eq.\@ \eqref{eq:spinortrafo}. In this way, we find that the components of $\psi$ transform as
\begin{equation} \label{TransformationEquationforWaveFunction}
\psi_{s_1 ... s_N}' (x_1,...,x_N) ~=~ \prod_{k=1}^{N} \left( \cosh \beta - s_k \sinh \beta \right) \psi_{s_1 ...s_N}(\Lambda^{-1}x_1,...,\Lambda^{-1}x_N).
\end{equation} 
This means that one obtains a factor of $(\cosh \beta - \sinh \beta)$ for every plus and a factor of $(\cosh \beta + \sinh \beta)$ for every minus in the index $(s_1...s_N)$. Hence, components with an equal number of plus and minus signs transform in the same way.

\begin{example}
We discuss the case $N=3$ in order to motivate the general form \eqref{eq:probconsbdc} of the boundary conditions. Consider a boundary point  $p = (t, z_1, t, z, t, z) \in \mathscr{C}_{1,t}$. We use eq.\@ \eqref{eq:EquationforPsiinComponentsthatOmegavanishes} to compute explicitly what the condition of probability conservation amounts to:
\begin{equation} \label{eq:ExampleEquation}
\omega_j(p)~ =~ 0~~ \Leftrightarrow~~ |\psi_{--+}|^2(p) - |\psi_{-+-}|^2(p) + |\psi_{+-+}|^2(p) -|\psi_{++-}|^2(p) ~=~ 0.
\end{equation}
Now we Lorentz transform this condition according to eq.\@ \eqref{TransformationEquationforWaveFunction} using the identity\\
 $\left( \cosh \beta - \sinh \beta \right) \left( \cosh \beta + \sinh \beta \right) = 1$:
\begin{equation}
0 = \left( \cosh \beta - \sinh \beta \right) \left(|\psi_{--+}|^2 - |\psi_{-+-}|^2 \right)(p') +  \left( \cosh \beta + \sinh \beta \right) \left( |\psi_{+-+}|^2 -|\psi_{++-}|^2 \right) (p'),
\end{equation}
where $p' = (\Lambda^{-1}(t,z_1),\Lambda^{-1}(t,z),\Lambda^{-1}(t,z))$. 
One can see that this equation cannot possibly be Lorentz invariant as a whole. Thus, demanding Lorentz invariance, we have to split up eq.\@ \eqref{eq:ExampleEquation} into two separate conditions relating only components which have the same number of plus and minus signs in their indices, i.e. 
\begin{equation}
|\psi_{-+-}|^2(p) - |\psi_{--+}|^2(p) ~=~ 0,~~~ \ |\psi_{+-+}|^2(p) -|\psi_{++-}|^2(p) ~=~ 0.
\end{equation}
These equations are equivalent to 
\begin{equation}
\psi_{--+}(p) ~=~ e^{i \varphi_-} \psi_{-+-}(p) ,~~~ \ \psi_{++-}(p)  ~=~ e^{i \varphi_+}\psi_{+-+}(p).
\end{equation}
A priori, the phases $\varphi_{\pm}$ could be functions of all particle coordinates. However, it is reasonable to demand invariance under Poincar\'{e} transformations. Then the phases $\varphi_{\pm}$ may only depend on the Minkowski distances of pairs of particles, $(x_i - x_j)^{\mu} (x_i - x_j)_{\mu}$. In the 3-particle case, there is only one such variable, $s^2 := (t_1 - t)^2 - (z_1 - z)^2$. However, the Minkowski distance $s^2$ changes along a multi-time characteristic. However, this would lead to a contradiction because the solution has to be constant along the characteristic (see lemma \ref{thm:generalsolution}). Thus, $\varphi_{\pm}$ must be constant in order for solutions to exist. A further investigations of the existence and uniqueness problem shows that even $\varphi_+ = \varphi_-$ is necessary for solutions to exist.\\
This yields the following general picture. Exchanging $s_k \leftrightarrow s_{k+1}$ in $\psi_{s_1...s_k s_{k+1}...s_N}$ on $\mathscr{C}_{k, k+1}$ only yields a phase factor which must not depend on the other spin indices (but may depend on $k$). These insights motivate the choice of boundary conditions \eqref{eq:probconsbdc}. The following lemma shows that this choice is indeed Lorentz invariant.
\end{example}

\begin{lemma} \label{thm:LemmaLorInvBC}
The probability conserving boundary conditions of lemma \ref{thm:MyBCareProbconserving}, i.e.
\begin{equation} \label{eq:probconsbdc2}
\psi_{s_1 ... s_{k-1} + - s_{k+2}... s_N} ~\hastobe~  e^{i \varphi^{(k)}} \psi_{s_1 ... s_{k-1} - + s_{k+2}... s_N} ~~\mathrm{on} ~~ \mathscr{C}_{k, k+1},~~k = 1,...,N-1
\end{equation}
are Lorentz invariant.
\end{lemma}

\begin{proof}
According to eq.\@ \eqref{TransformationEquationforWaveFunction}, eq.\@ \eqref{eq:probconsbdc2} has the same form in every Lorentz frame. Besides, the sets $\mathscr{C}_{k,k+1}$ on which the condition is prescribed are Lorentz invariant.  \qed
\end{proof}

\begin{remark} One may ask if there are other possible choices of boundary conditions which lead to $\omega_j=0$ at the boundary and are Lorentz invariant. The example shows that for $N=3$ we have already found the only one. For $N\geq4$ there may exist more complicated boundary conditions with the desired properties. However, aiming at a model valid for all $N \geq 2$, we do not further pursue this question here.
\end{remark}

\section{Existence and uniqueness of solutions} \label{sec:SectionExistenceandUniquenessSpacelikeConf}

We now come to the main result of the paper: the theorem on the existence and uniqueness of solutions for the boundary conditions discussed so far (thm.\@ \ref{thm:THETheorem}). Furthermore, we find an explicit formula for the unique solution of a given IBVP within the class defined by eq.\@ \eqref{eq:model} with $\Omega = \mathscr{S}_1$ and boundary conditions \eqref{eq:probconsbdc}.\\
We start with providing some intuition about the behaviour of the solution for $N = 3$. The main idea is to make use of the fact that the components of the solution have to be constant along the multi-time characteristics \cite{1d_model}.

\begin{example}
For $N = 3$ the wave function has $2^3 = 8$ components. According to lemma \ref{thm:generalsolution}, these are constant along their respective multi-time characteristics. We visualize the multi-time characteristics as follows (see fig.\@ \ref{fig:FirstPicture}). One can see from eq.\@ \eqref{eq:multitimecharacteristic} that the multi-time characteristics are the Cartesian product of $N=3$ lines. These lines are plotted in the same space-time diagram. Any combination of three points on the different lines constitutes an element of the respective multi-time characteristic. The slopes of the various lines characterize the associated component $\psi_{s_1...s_N}$. More precisely, a line with nevative (positive) slope for particle $k$ is associated with the appearance of $s_k = +1$ ($s_k = -1$) in the index of $\psi$.\\
Fig.\@ \ref{fig:FirstPicture} shows a multi-time characteristic $S_{+++}(c_1,c_2,c_3)$ for the component $\psi_{+++}$ where the $c_k$ are defined by a certain point $p=(A,B,C) \in \mathscr{S}_1$ (see eq.\@ \eqref{eq:characteristicvalues}). $\psi_{+++}$ is determined at $p$ by initial data because the whole characteristic $S_{+++}(c_1,c_2,c_3)$ is contained in $\mathscr{S}_1$. This can be seen from the fact that every three points on different lines are space-like related. Besides, the value $\psi(p)$ is determined \textit{uniquely} by initial data as there exists a unique intersection point of $S_{+++}(c_1,c_2,c_3)$ with the surface $t_1 = t_2 = t_3 = 0$, given by $(0, c_1, 0, c_2, 0, c_3)$. Thus, we obtain
\begin{equation}
\psi_{+++} (t_1, z_1, t_2, z_2, t_3, z_3) ~=~ g_{+++} (0, c_1, 0, c_2, 0, c_3),
\end{equation}
where $g_{+++}$ are the initial values for $\psi_{+++}$.

\begin{figure}[h]
\begin{center}
\includegraphics[scale=0.8]{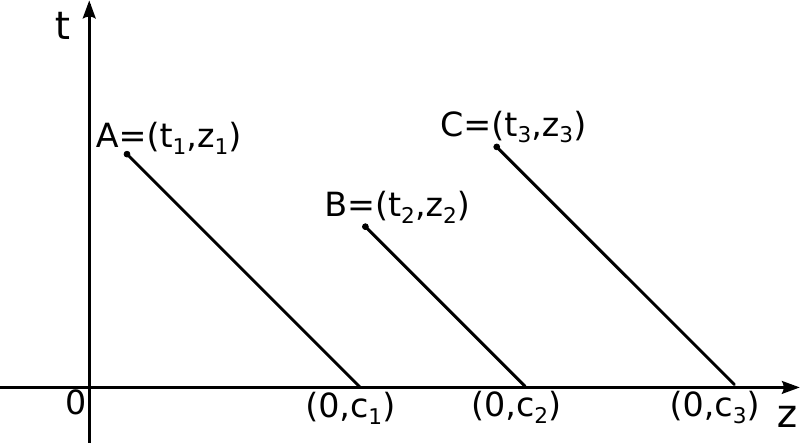}
\caption{\small A multi-time characteristic $S_{+++}(c_1,c_2,c_3)$ for the component $\psi_{+++}$. $S_{+++}(c_1,c_2,c_3)$ is the Cartesian product of three lines which are plotted in the same space-time diagram. Every triple of points on different lines, e.g.\@ $(A,B,C)$, is contained in $\mathscr{S}_1$.}
\label{fig:FirstPicture}
\end{center}
\end{figure}

For a component $\psi_{s_1 s_2 s_3}$ containing plus as well as minus signs in the index, for example $\psi_{+-+}$, the situation is different (see fig.\@ \ref{fig:SecondPicture}). One can see that intersections of the lines defining a characteristic $S_{+-+}$ do occur in the diagram.
\begin{figure}[h] 
\begin{center}
\includegraphics[scale=0.8]{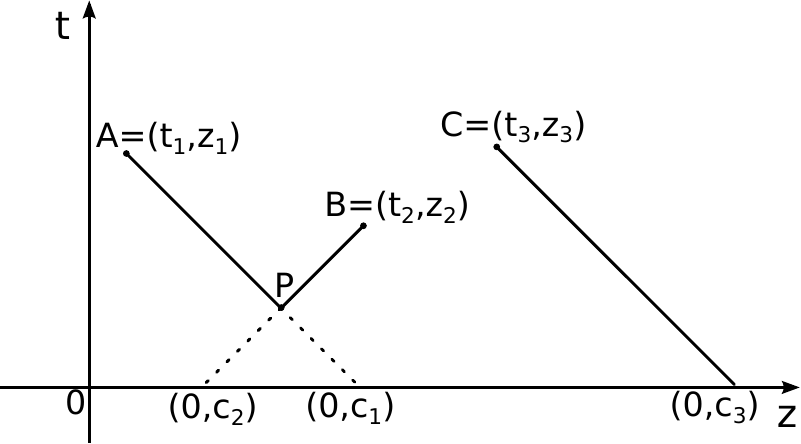}
\caption{\small A multi-time characteristic $S_{+-+}$ for $\psi_{+-+}$, depicted for the same configuration as in fig.\@ \ref{fig:FirstPicture}. One cannot trace back the lines to the initial data surface $\mathcal{I}$ because one leaves $\mathscr{S}_1$ at the point $P$. Instead one has to first make use of the boundary conditions and can trace back the multi-time characteristic for $\psi_{-++}$ which corresponds to the same lines but with particle labels 1 and 2 exchanged.
}
\label{fig:SecondPicture}
\end{center}
\end{figure}
\\ When an intersection occurs, the multi-time characteristic leaves $\mathscr{S}_1$. Therefore, a situation like in fig.\@ \ref{fig:SecondPicture} can occur: tracing back the multi-time characteristic to the initial data surface, one leaves the domain. Thus, $\psi_{+-+}(A,B,C)$ is not defined solely by initial values. To obtain the value of  $\psi_{+-+}(A,B,C)$, we first realize $\psi_{+-+} (A, B, C) = \psi_{+-+} (P, P, C)$. Then we employ the boundary conditions to obtain $ \psi_{+-+} (A, B, C) = \psi_{+-+} (P, P, C)  = e^{i \varphi^{(1)}}\psi_{-++} (P, P, C)$. The component $\psi_{-++}$ is now determined at $(P,P,C)$ by initial data in a similar way as before, i.e.\@ $\psi_{-++} (P, P, C) = g_{-++}(0, c_2, 0, c_1, 0, c_3)$. Summarizing the relations, we obtain:
\begin{equation}
\psi_{+-+} (t_1, z_1, t_2, z_2, t_3, z_3) ~=~  e^{i \varphi^{(1)}} g_{-++} (0, c_2, 0, c_1, 0, c_3).
\end{equation}  
Pictorially speaking, this amounts to exchanging particle labels and picking up a phase while leaving the domain on the way back to the initial data surface.

The considerations above hint at a general idea: it is possible to obtain an explicit formula for the solutions of the IBVP (see eq.\@ \eqref{eq:Formulaforsolutions}) by a process of successively tracing back components to collisions, using the boundary conditions to switch the component, tracing back to the next collision and finally arriving at the initial data. In this way, one can also determine values of components with multiple intersections of the lines constituting the multi-time characteristic, like in fig.\@ \ref{fig:ThirdPicture}. This motivates the theorem below.

\begin{figure}[h] 
\begin{center}
\includegraphics[scale=0.8]{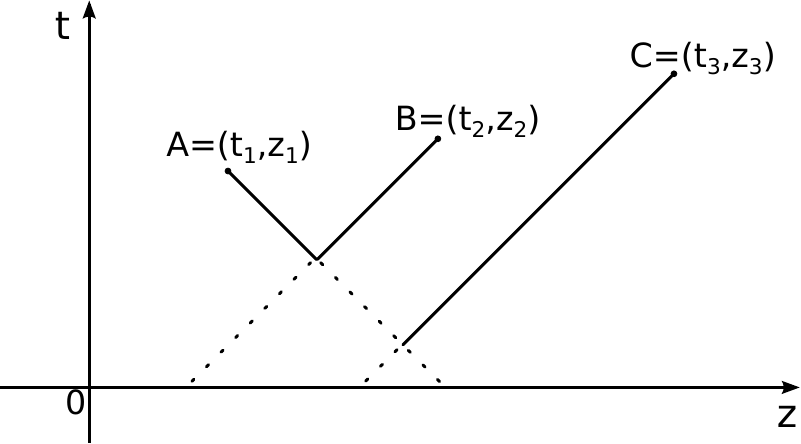}
\caption{\small A multi-time characteristic with several intersections for the component $\psi_{+--}$.}
\label{fig:ThirdPicture}
\end{center}
\end{figure}    

\end{example}

\begin{theorem}
\label{thm:THETheorem}
Let $m \in \mathbb{N}$ and choose initial data $g_j \in  C^m \left( \mathcal{I} \cap  \overline{\mathscr{S}}_1 , \mathbb{C}  \right) \ \forall j = 1, ... , 2^N$ such that they also satisfy the boundary conditions, i.e.
\begin{align} \label{eq:BoundaryValuesforInitialConditions}
 & g_{s_1 ... s_{k-1}+- s_{k+2} ... s_N}  ~=~   e^{i \varphi^{(k)}} g_{s_1 ... s_{k-1} -+ s_{k+2} ... s_N} ~~ {\rm on}~ \mathcal{I} \cap \mathscr{C}_{k,k+1}, ~~k=1,...,N-1
 \end{align}
and let this transition be $C^m$.\\
 Then there exists a unique solution $\psi \in C^m\left( \overline{\mathscr{S}}_1 , (\C^2)^{\otimes N}  \right)$ of the IBVP \eqref{eq:model} with $\Omega = \mathscr{S}_1$, boundary conditions \eqref{eq:probconsbdc} and initial data
 \begin{equation}
  \psi_{s_1...s_N}(0,z_1,...,0,z_N)~=~ g_{s_1...s_N}(z_1,...,s_N),~~z_1 \leq \dots \leq z_N.
  \label{eq:initialvalues}
 \end{equation}
 If all characteristic values $c_k = z_k + s_k t_k$ are different, the components of the solution are explicitly given by
\begin{equation} \label{eq:Formulaforsolutions}
\psi_{s_1 ... s_N} (t_1, z_1, ..., t_N, z_N) ~=~ e^{i \phi^{\pi}_{s_1 ... s_N} } g_ {s_{\pi(1)} ... s_{\pi(N)}} \left( c_{\pi(1)}, ... , c_{\pi(N)} \right),
\end{equation} 
where $\pi$ is the permutation with
\begin{equation}
c_{\pi(1)} < \dots < c_{\pi(N)}.
 \label{eq:ordering}
\end{equation} 
$\phi^{\pi}_{s_1... s_N}$ is the phase which is uniquely determined by the $\varphi^{(k)}$ of eq.\@ \eqref{eq:probconsbdc} via the definition below.\\
If some of the $c_k$ are equal, $\psi_{s_1... s_N}$ is given by the continuation of eq.\@ \eqref{eq:Formulaforsolutions}.\\ 
In addition, the model ensures probability conservation on general space-like hypersurfaces and is Lorentz invariant.
\end{theorem}

\begin{definition}
A pair $(k,l) \in \left\lbrace 1, ... , N \right\rbrace^2$ with $k<l$ is said to be a \textit{collision} of a transposition $\pi \in S_N$ iff $\pi(k) > \pi(l)$. 
\end{definition}

\begin{definition}
The phases $\phi^{\pi}_{s_1 ... s_N}$ appearing in thm.\@ \ref{thm:THETheorem} are given by the conditions:
\begin{enumerate}
 \item $\phi^{\rm id}_{s_1 ... s_N} = 0$,
 \item Let $\tau_k$ be the transposition of $k$ and $k+1$. If $\pi$ can be decomposed as $\pi = \tau_k \circ \sigma$ where $\sigma$ is a permutation with fewer collisions than $\pi$, then
\begin{equation}
\label{eq:PhaseCondition}
\phi^{\pi}_{s_1 ... s_N} ~=~ \phi^{\sigma}_{s_1 ... s_{k+1} s_k ... s_N} + s_k  \varphi^{(k)}.
\end{equation}
\end{enumerate}
\end{definition}

\begin{lemma}
 The phases $\phi^{\pi}_{s_1 ...s_N}$ exist and are uniquely determined.
 \label{thm:phiwelldefined}
\end{lemma}

\begin{proof}
We proceed via induction over the number of collisions.\\
\textit{Induction start:} If $\pi$ has no collision, $\pi = \mathrm{id}$. Thus, the phase $\phi^{\rm id}_{s_1 ... s_N}$ is determined uniquely by the first condition in the definition. If $\pi$ has exactly one collision, then it is just a transposition of neighbouring elements, so there exists some $k$ with $\pi = \tau_k = \tau_k \circ \mathrm{id}$ and the phase is uniquely determined by \eqref{eq:PhaseCondition} as $\phi^{\pi}_{s_1 ... s_N} = s_k \varphi^{(k)}$.\\
\textit{Induction step:} Assume that all phases $\phi^{\pi' }_{s_1 ... s_N}$ for permutations with $n\geq 1$ collisions are uniquely determined and let $\pi$ have $n+1$ collisions.
 It is known from the general theory of permutations that there exists at least one permutation $\sigma$ with $n$ collisions and a neighbouring transposition $\tau_s$ such that $\pi = \tau_s \circ \sigma$. However, it may be possible to decompose $\pi$ in two different ways:
\begin{equation}
\pi ~=~ \tau_s \circ \sigma ~=~ \tau_k \circ \kappa
\label{eq:twodifferentways}
\end{equation}
where $s, k \in \left\lbrace 1, \dots , N \right\rbrace$, $s \neq k$ and $\sigma, \kappa$ are permutations with at least $n$ collisions. In order for these two permutations to have one collision less than $\pi$, we see that $(k, k+1)$ and $(s, s+1)$ must be collisions of $\pi$.\\
To show that despite the different ways of decomposition, the corresponding phases are uniquely defined, we make use of the fact that the phases need only be defined for a certain type of permutation. To characterize them, we prove an auxiliary claim: in the above situation, $\tau_s$ commutes with $\tau_k$ because $|s - k| \neq 1$.
\\ \par \begingroup
\leftskip=15pt \textit{Claim:} Let $\psi_{s_1 \dots s_N}$ and $(t_1, z_1, \dots , t_N, z_N) \in \mathscr{S}_1$ such that there is a collision, i.e.\@ a pair $(a, b)$ with $a<b$ and $c_a > c_b$. Then one of the following two possibilities holds: 
 \begin{equation}
\left\lbrace
 \begin{array}{c}
 \mathrm{either} \ s_a = +1 \ \wedge\  s_b = -1 \ \wedge \ t_a>0 \ \wedge \ t_b>0  
 \\ \hspace{0.35cm} \mathrm{or} \hspace{0.35cm} s_a = -1 \ \wedge \ s_b = +1 \ \wedge \ t_a<0 \ \wedge \ t_b<0 
 \end{array} \right.
\end{equation}  
\\ \textit{Proof of the Claim:}
We know that $a<b$, $c_a > c_b$ and $z_a < z_b$. We show that  $s_a = +1$ implies $s_b = -1 \ \wedge \ t_a>0 \ \wedge \ t_b>0$; the second case follows analogously. \\
If $s_a = +1$, then 
 \begin{flalign} 
 c_a & ~>~  c_b \nonumber
 \\ \Leftrightarrow~~~ z_a + t_a & ~>~  z_b +s_b t_b \nonumber
 \\ \Leftrightarrow~~~  t_a - (s_b t_b) & ~>~ z_b - z_a ~=~ | z_b - z_a |.
 \end{flalign}
 If now $s_b = +1$, this would be a contradiction to the points $(t_a,z_a)$ and $(t_b,z_b)$ being space-like separated. Hence, $s_b = -1$, so we have:
 \begin{equation}
 |z_b - z_a| ~<~ t_a + t_b ~=~ |t_a + t_b|.
 \end{equation}
 This implies that $t_a$ and $t_b$ cannot both be negative. So assume that one of them is negative, w.l.o.g. $t_a>0, t_b<0$. But then $|t_a - t_b| > |t_a + t_b|$ and
 \begin{equation}
 |z_k - z_a| ~<~ |t_a + t_b| ~<~ |t_a -t_b|,
 \end{equation}
  which also is a contradiction to the points being space-like. Thus, one must have $t_a>0, t_b>0$, which proves the claim. 
\par 
\endgroup
~\\
Because of the specific sign combinations that allow for collisions, the claim shows that if $(s, s+1)$ is a collision, neither $(s-1, s)$ nor $(s+1, s+2)$ can be one. Therefore, $|k - s| \geq 2$ and $\tau_k, \tau_s$ commute.\\
We use the commutability of $\tau_k$ and $\tau_s$ to define a third permutation
\begin{equation}
\rho ~:=~ \tau_k \circ \tau_s \circ \pi ~=~ \tau_k \circ \sigma ~=~ \tau_s \circ \kappa,
\end{equation}
which by construction has $n-1$ collisions, i.e.\@ one less than $\sigma$ and $\kappa$. This means the seemingly different representations of $\phi^{\pi}_{s_1 ... s_N}$,
\begin{align}
\phi^{\pi}_{s_1 \dots s_N} & ~=~  \phi^{\sigma}_{s_1 \dots s_{s+1} s_s \dots s_N} + s_s  \varphi^{(s)} ~~~ \mathrm{and}  \nonumber
\\ \phi^{\pi}_{s_1 \dots s_N} & ~=~  \phi^{\kappa}_{s_1 \dots s_{k+1} s_k \dots s_N} + s_k  \varphi^{(k)} 
\end{align} 
are in fact equal. This can be seen from the fact that the different ways of decomposing $\pi$ via eq.\@ \eqref{eq:twodifferentways} yield (using \eqref{eq:PhaseCondition}):
\begin{align}
 \phi^{\pi}_{s_1 ... s_N} ~=~\phi^{\sigma}_{s_1 \dots s_{s+1} s_s \dots s_N} + s_s  \varphi^{(s)} & ~=~ \phi^{\rho}_{s_1 \dots s_{k+1} s_k \dots s_{s+1} s_s \dots s_N} + s_k \varphi^{(k)} +s_s  \varphi^{(s)} \nonumber
 \\ & ~=~ \phi^{\kappa}_{s_1 \dots s_{k+1} s_k \dots s_N} + s_k \varphi^{(k)} ~=~\phi^{\pi}_{s_1 ... s_N}.
\end{align}
This finishes the proof of uniqueness of the phases because by the induction hypothesis, the phases associated with $\rho$, $\sigma$ and $\kappa$ exist and are uniquely determined.\qed 
\end{proof}

\begin{proofT}
The points of Lorentz invariance and probability conservation are clear from lemma \ref{thm:LemmaLorInvBC} and lemma \ref{thm:MyBCareProbconserving}, respectively. Furthermore, we already know that the uniqueness of solutions in a weak sense follows from probability conservation by virtue of thm.\@ \ref{thm:ProbgivesUniquenessTheorem}. If the function defined by eq.\@ \eqref{eq:Formulaforsolutions} is indeed $m$ times continuously differentiable, it follows from continuity that it is also unique as a $C^m$-function.\\
Thus, it only remains to show that the function given by \eqref{eq:Formulaforsolutions} is indeed a classical solution of the IBVP. In order to prove this, the following four points have to be verified:
\begin{enumerate}
\item \textit{Differentiability:} We need to prove that $\psi \in C^m \left( \mathscr{S}_1 , (\mathbb{C}^2)^{\otimes N} \right)$. As the initial values satisfy $g_j \in  C^m \left( \mathcal{I} , \mathbb{C}   \right) \ \forall j = 1, ... , 2^N$, this  property is inherited by $\psi_j$ via the characteristics. To see this, note that eq.\@ \eqref{eq:Formulaforsolutions} just makes use of a translation of the initial values along straight lines in the $(t_k,z_k)$ spaces.\\
However, we need to consider those points separately where the permutation $\pi$ changes. This exactly happens when at least two of the characteristic values $c_j$ are equal. But then the $C^m$-property of $\psi$ is assured by the requirement that the initial values must satisfy the boundary conditions (eq.\@ \eqref{eq:BoundaryValuesforInitialConditions}) and that the transition shall be $C^m$.
\item The function defined by eq.\@ \textit{\eqref{eq:Formulaforsolutions} solves the system of Dirac equations in $\mathscr{S}_1$:} This follows from lemma \ref{thm:generalsolution} because the components of the solution are indeed constant along the respective multi-time characteristics and only depend on the characteristic values $c_k$.
\item \textit{The initial conditions \eqref{eq:initialvalues} are satisfied:} At a point $(0, z_1, 0, z_2, ... , 0, z_N) \in \mathcal{I} \cap \overline{\mathscr{S}}_1$, we have $c_k = z_k \ \forall k$ and thus $c_{\pi(1)} \leq c_{\pi(2)} \leq \dots \leq c_{\pi(N)}$ is fulfilled for $\pi = {\rm id}$. Therefore, formula \eqref{eq:Formulaforsolutions} reduces to
\begin{equation}
\psi_{s_1 ... s_N} (0, z_1, ..., 0, z_N) ~=~ g_ {s_{1} ... s_{N}} \left( c_{1}, ..., c_{N} \right)
\end{equation}
which is equivalent to \eqref{eq:initialvalues}.
\item \textit{The boundary conditions \eqref{eq:probconsbdc} are satisfied:} 
Let $k \in \left\lbrace 1, ..., N-1 \right\rbrace$ and $(t_1, z_1, ..., t_k = t, z_k = z, t_{k+1} = t, z_{k+1} = z, ..., t_N, z_N) \in \mathscr{C}_{k, k+1}$. We consider two components of $\psi$ where only the $k$-th and $(k+1)$-th sign is exchanged, or more formally: Let $(s_1, \dots s_N)$, $(\tilde{s}_1, \dots \tilde{s}_N) \in \left\lbrace -1, +1 \right\rbrace^N$ with $s_l = \tilde{s}_l \ \forall l \notin \left\lbrace k, k+1 \right\rbrace$ and $( s_k, s_{k +1}) = ( +, -) = ( \tilde{s}_{k+1}, \tilde{s}_{k})$. The respective characteristic values are given by $c_k = z_k +s_k t_k$  and $\tilde{c}_k = z_k +\tilde{s}_k t_k$.
\newline Now observe the property  $c_l = \tilde{c}_l \ \forall l \notin \left\lbrace k, k+1 \right\rbrace$ and $c_k = \tilde{c}_{k +1}$, $\tilde{c}_{k+1} = c_k.$ Let $\pi$ be the permutation that leads to $c_{\pi(1)} \leq \dots \leq c_{\pi(N)}$. The permutation $\sigma$ needed to achieve $\tilde{c}_{\sigma(1)} \leq  \dots \leq \tilde{c}_{\sigma(N)}$  is given by $\sigma = \tau_k \circ \pi$, and it has one collision less than $\pi$ with respect to the indices $\tilde{s}_k$. Inserting this into eq.\@ \eqref{eq:Formulaforsolutions} yields
\begin{flalign} 
\psi_{\tilde{s}_1 \dots \tilde{s}_N} & \overset{\eqref{eq:Formulaforsolutions}}{=} e^{i \phi^{\sigma}_{\tilde{s}_1 \dots \tilde{s}_N}} g_ {\tilde{s}_{\sigma(1)}  \dots \tilde{s}_{\sigma(N)}} \left( \tilde{c}_{\sigma(1)}, ... , \tilde{c}_{\sigma(N)} \right) & 
\nonumber\\ & \overset{\eqref{eq:PhaseCondition}}{=} e^{i \phi^{\pi}_{s_1 \dots s_N}} e^{-i \varphi^{(k)} } g_ {s_{\pi(1)}  \dots s_{\pi(N)}} \left( c_{\pi(1)}, ..., c_{\pi(N)} \right)  & 
\nonumber\\ & \overset{\eqref{eq:Formulaforsolutions}}{=} e^{-i \varphi^{(k)}} \psi_{s_1 \dots s_N} . & 
\end{flalign}
This shows that \eqref{eq:probconsbdc} is valid. 
\end{enumerate}
These four points establish existence; the function given in \eqref{eq:Formulaforsolutions} is indeed the solution of the IBVP.\qed
\end{proofT}

\begin{remark}
\begin{enumerate}
 \item Uniqueness of solutions can also be proven differently than by invoking thm.\@ \ref{thm:ProbgivesUniquenessTheorem}, namely by directly showing that every solution of the IBVP has to fulfil eq.\@ \eqref{eq:Formulaforsolutions}. A possible proof is via induction over the number of collisions.
 \item Note that on purely dimensional grounds it is remarkable that solutions for the IBVP with boundary conditions \eqref{eq:probconsbdc} do exist. Because the dimension of $\mathscr{C}_{k,k+1}$ is $(N-1)(1+d)$ which for $N>1+d$ is greater than $Nd$, the dimension of the initial data surface $\mathcal{I}\cap \mathscr{S}_1$, one might have suspected this not to be the case.
\end{enumerate}

\end{remark}


\section{Interaction and effective potential} \label{sec:interaction}

In addition to the mathematical and physical features already established, we now prove in this section that our model is interacting. Moreover, we outline how the interaction can be described effectively at equal times using self-adjoint extensions of the free two-particle Dirac Hamiltonian.

In \cite{1d_model}, a general criterion for interaction was given. We call a physical model \textit{interacting} iff it generates entanglement, i.e.\@ if there exist wave functions that are initially product states and become entangled during time evolution. Note that for the antisymmetrized wave functions we are considering, a product means ``wedge product''.
We now present a simple argument why our model is interacting in this sense. 

\begin{lemma} \label{thm:Interactionlemma}
The model defined by \eqref{eq:model} with $\Omega = \mathscr{S}_1$ and boundary conditions \eqref{eq:probconsbdc} is interacting if there exists $k \in \left\lbrace 1, \dots, N-1 \right\rbrace$ with $\varphi^{(k)} \neq \pi$.
\end{lemma}

\begin{proof}
W.l.o.g.\@ $k=1$.
Let the initial conditions be such that $\left. \psi \right|_{\mathcal{I}}$ is a product wave function. In particular, this means that there exist functions $\alpha, \beta, \gamma, \delta \in C^m (\R, \C) $ and  $\zeta \in C^m (\R^{N-2}, \C) $  with
\begin{align}
\nonumber g_{+-+\dots+}(z_1,...,z_N) ~&=~ \alpha (z_1) \beta (z_2) \zeta (z_3, ..., z_N) \\
g_{-++\dots+}(z_1,...,z_N) ~&=~ \gamma (z_1) \delta (z_2) \zeta (z_3, ..., z_N)
\end{align}
for $z_1 \leq \dots \leq z_N$.\\
Antisymmetry \eqref{eq:antisymmetrycond} implies
\begin{equation}
\alpha (z_1) \beta (z_2)~=~  - \gamma (z_2) \delta (z_1) .
\label{eq:concreteexampleisantisymm}
\end{equation}
 Consider the solution at a point $p=(t, z_1, ..., t ,z_N) \in \mathscr{S}_1$ with common time $t>0$. The auxiliary claim in the proof of lemma \ref{thm:phiwelldefined} implies that the characteristic values at $p$ with respect to the component $\psi_{+-+\dots+}$ are in ascending order iff $z_1 + t \leq z_2 - t$. Thus we can use formula \eqref{eq:Formulaforsolutions} to obtain $\psi_{+-+\dots+}(p)$, with the permutation $\pi$ being the identity if $z_1 \leq z_2 - 2t$ and the transposition $\tau_1$ if $z_1 > z_2 - 2t$. Written via the Heaviside function $\Theta$, this yields
\begin{flalign} \nonumber
\psi_{+-+\dots+}(p) ~=~ &~ g_{+-+\dots+} (c_1, ... c_N) \, \Theta\left( z_2 - z_1 - 2t \right)
\\ \nonumber & + e^{ i \varphi^{(1)}} g_{-++\dots+} (c_2, c_1, c_3, ...,c_N) \, \Theta\left( 2t + z_1 - z_2 \right)
\\ \nonumber ~=~ & ~\alpha (c_1) \beta (c_2) \zeta (c_3, ..., c_N) \, \Theta\left( z_2 - z_1 - 2t \right)
\\ & + e^{ i \varphi^{(1)}} \gamma(c_2) \delta(c_1) \zeta (c_3, ..., c_N)  \, \Theta\left( 2t + z_1 - z_2 \right).
\label{eq:ComponentnotProduct}
\end{flalign}   
Using \eqref{eq:concreteexampleisantisymm}, the expression simplifies to
\begin{equation}
\psi_{+-+\dots+}(p) ~=~ \alpha (c_1) \beta (c_2) \zeta (c_3, \dots , c_N) \left( \Theta (c_2 - c_1) - e^{i \varphi^{(1)}} \Theta (c_1 - c_2)  \right).
\end{equation}
This function contains the Heaviside function of a combination of $t$, $z_1$ and $z_2$ in a non-factorizable way. The $\Theta$-function cannot be left away for general initial values (as might be the case if they were zero in some regions). Furthermore, because of the prefactor $e^{ i \varphi^{(1)}}$ of the second summand, we cannot write it as a product as long as $\varphi^{(1)} \neq \pi$.
\qed
\end{proof}

\paragraph{Effective single-time model:} In the following, we show how an effective single-time model can be obtained from our model when considered at equal times $t_1 = t_2 =t$. Even though manifest Lorentz invariance is lost for a single-time model, we consider it instructive to connect with this familiar setting.\\
For simplicity, let $N=2$. We denote the single-time wave function by $\chi(z_1,z_2,t) := \psi(t,z_1,t,z_2)$. Then the single-time model is given by the domain $\{ (z_1,z_2,t) \in \R^3 : z_1 \neq z_2 \}$, initial data at $t = 0$, boundary conditions \eqref{eq:probconsbdc} (with $t_1,t_2$ replaced by $t$ in all the constructions) and the wave equation
\begin{equation}
i \frac{\partial \chi}{\partial t} ~=~  -i \left( \sigma_3 \otimes \id \frac{\partial}{\partial z_1} + \id \otimes \sigma_3 \frac{\partial}{\partial z_2} \right) \chi ~\equiv~ \hat{H} \chi.
\label{eq:singletimedirac}
\end{equation}
Note that eq.\@ \eqref{eq:singletimedirac} is obtained from the multi-time equations \eqref{eq:multitimedirac} by the chain rule.\\
We introduce new coordinates $u = \frac{1}{2} (z_1 - z_2)$ and $v = \frac{1}{2} (z_1 + z_2)$. The Hamiltonian becomes
\begin{equation}
\hat{H} ~=~ -i \, \diag(\partial_v, \partial_u, - \partial_u, - \partial_v).
\label{eq:h0diag}
\end{equation}
The boundary condition \eqref{eq:probconsbdc} can be reformulated using antisymmetry, i.e.\@ $\chi_2 (u, v,t) = - \chi_3 (-u, v,t)$, with the result
\begin{align} \label{eq:BCPotential}
\lim_{u \nearrow 0} \chi_2(u,v,t) & ~=~ \lim_{u  \searrow 0} -e^{-i \varphi} \chi_2 (u, v,t),~~t \in \R, \nonumber \\ 
\lim_{u \nearrow 0} \chi_3(u,v,t) & ~=~ \lim_{u  \searrow 0} -e^{i \varphi} \chi_3 (u, v,t),~~\,~ t \in \R.
\end{align}
The components $\chi_1$ and $\chi_4$ evolve freely and have to be continuous and zero at $u = 0$ because of antisymmetry. For $\varphi = \pi$, eq.\@ \eqref{eq:BCPotential} also reduces to the condition of continuity. In that case, the model becomes free -- in agreement with lemma \ref{thm:Interactionlemma}.
\\ The boundary conditions \eqref{eq:BCPotential} can be implemented in a rigorous functional-analytic way by a family of self-adjoint extensions of $\hat{H}$ parametrized by $\varphi$ \cite{LukasThesis}. Moreover, the unitary groups generated by these self-adjoint extensions can also be obtained as the limits of unitary groups generated by Hamiltonians $\hat{H}_n = \hat{H} + V_n$ with
\begin{equation}
 V_n(u) ~=~ \diag \left( 0, V_n(u), -V_n(u), 0 \right),
\label{eq:v}
\end{equation}
where the potentials converge to the $\delta$-function, $V_n(u) \rightarrow (\pi-\varphi)\delta(u)$, in the distributional sense. In essence, this is what the physics literature shows \cite{kellarstephenson,diracpotential}\footnote{Note that \cite{kellarstephenson,diracpotential} consider a one-particle Dirac equation. A comparison with these papers is never\-the\-less possible since, as evident from eqs.\@ \eqref{eq:h0diag}, \eqref{eq:v} the single-time equation \eqref{eq:singletimedirac} with additional potential \eqref{eq:v} decouples and $\tilde{\chi} := (\chi_2,\chi_3)$ satisfies a one-particle Dirac equation in $u$, i.e. $ i \partial_t \tilde{\chi}(u,v) = \left[ -i \sigma_3 \partial_u + \sigma_3 V_n(u) \right] \tilde{\chi}(u,v)$.}.
Recalling $u = \frac{1}{2}(z_1-z_2)$, our model can therefore be considered a relativistic version of the $1+1$-dimensional multi-particle Dirac equation with a spin-dependent $\delta(z_1 - z_2)$-potential for $\varphi \neq \pi$.


\section{Non-existence of solutions for configurations with a minimal space-like distance} \label{sec:spacelikeconfigs}

In this section, we leave the setting of the model \eqref{eq:model}, just retaining the multi-time Dirac equation \eqref{eq:multitimedirac}, and focus on a different aspect: the question whether a consistent Lorentz-invariant and probability-conserving dynamics exists on the domain $\mathscr{S}_\alpha$ of space-like configurations with a minimum space-like distance $\alpha$.\\
The relevance of the question is motivated as follows: one may ask whether an immediate generalization of the model \eqref{eq:model} to higher dimensions is possible which still produces interaction. This question was answered negatively in \cite{1d_model}, the reason being that for $d > 1$ the probability flux across the boundary vanishes without the need for boundary conditions and thus, by theorem \ref{thm:ProbgivesUniquenessTheorem}, uniqueness of solutions immediately holds. The model then corresponds to the free case. The main point is that for $d > 1$, the set of coincidence points in space-time is too low-dimensional to have impact on the dynamics. Therefore, it suggests itself to ask whether a change in the domain (and thereby of its boundary) can be made such that the dimension of the set across which the probability flux could leave the boundary is increased. A natural choice is the set of $\alpha$-\textit{space-like configurations} (here for $N = 2$):
\begin{equation}
  \mathscr{S}_{\alpha} ~=~ \left\lbrace (t_1,  \mathbf{x}_1, t_2, \mathbf{x}_2) \in \R^{1+d} \times \R^{1+d} : (t_1 - t_2)^2 - (\mathbf{x}_1 - \mathbf{x}_2)^2 < -\alpha^2 \right\rbrace . 
\end{equation} 
$\partial \mathscr{S}_{\alpha}$ has dimension $2d + 1$ and its intersection with $\Sigma \times \Sigma$, the set appearing in the proof of probability conservation \cite{1d_model}, has dimension $2d$ which is sufficient to have impact on the dynamics (the reason being that $\omega_j$ is a $2d$-form). Compare with the dimension $1+d$ of the set of coincidence points.  However, because $\partial \mathscr{S}_{\alpha}$ itself has a dimension greater than $2d$, the dimension of the initial data surface, the question arises if there is a consistent dynamics on it at all. In the following, we approach this question for the simplest case, $d =1$, for which the necessary mathematical tools are available, and show that the answer is negative. 
\\ First we show that there can only be one kind of boundary conditions with the desired properties. In a second step we then prove that the corresponding IBVP on $\mathscr{S}_{\alpha}$ does not possess non-trivial solutions. We make use of the following definitions:
\begin{align}
\mathscr{S}^{+}_\alpha ~=~  \left\lbrace (t, z_1, t, z_2) \in  \mathscr{S}_{\alpha}: z_1 - z_2 > 0 \right\rbrace,
\nonumber
\\
\mathscr{S}^{-}_\alpha ~=~ \left\lbrace (t, z_1, t, z_2) \in  \mathscr{S}_{\alpha}: z_1 - z_2 < 0\right\rbrace.
\end{align}
We have: $\mathscr{S}_\alpha = \mathscr{S}^{+}_\alpha \cup \mathscr{S}^{-}_\alpha$.

\begin{lemma}
\label{thm:LemmaAlphaSpacelikePossibleBC}
Let $\alpha > 0$ and $N=2$. For the multi-time Dirac equations \eqref{eq:multitimedirac} on the domain $\mathscr{S}_{\alpha}$, there exist no other Poincar\'e invariant boundary conditions which lead to probability conservation on every space-like hypersurface and which are compatible with antisymmetry, than the ones given by:
\begin{equation}
\label{eq:GoodBCforSalpha}
\psi_{ + - }(p) ~=~  e^{\pm i \varphi}  \psi_{- + }(p) ~~\forall p \in  \partial \mathscr{S}^{\pm}_\alpha
\end{equation}
with a fixed $\varphi \in (-\pi,\pi]$. 
\end{lemma}

\begin{proof}

Let $p = (t_1, z_1, t_2, z_2) \in \partial \mathscr{S}_{\alpha}$. Because the two points $(t_1, z_1)$ and $(t_2, z_2)$ are space-like separated, there is a Lorentz frame with $t_1 = t_2$. We work in this frame, so we can write either $p = (t_p, z, t_p, z + \alpha)$ or $p = (t_p, z, t_p, z - \alpha)$. The idea is to use Stokes' theorem in a similar way as in the proof of lemma \ref{thm:probcons} to obtain a condition for probability conservation on equal-time hypersurfaces $\Sigma_{\tau_1}, \Sigma_{\tau_2}$ in the considered Lorentz frame. Here, w.l.o.g.\@ $\tau_1 < \tau_2$. Let
\begin{equation} 
V ~:=~ \left\lbrace (t, z_1, t, z_2) \in \overline{\mathscr{S}}_{\alpha} : \tau_1 \leq t \leq \tau_2 \right\rbrace.
\end{equation}
$V$ plays the same role as $V_R$ in eq.\@ \eqref{eq:StokesVolumen} for $R$ sufficiently large. Following the strategy of the proof of lemma \ref{thm:probcons}, one deduces
\begin{equation}
0 ~=~ \int_{V} d \omega_j ~=~  \int_{\partial V} \omega_j 
\end{equation}
Note that in contrast to the proof of lemma \ref{thm:probcons} but similar to \cite[proof of thm.\@ 4.4]{1d_model}, there now exist two connected components of the domain $\mathscr{S}_\alpha$. Therefore, probability conservation in the form
\begin{equation}
 \int_{(\Sigma_{\tau_1} \times \Sigma_{\tau_1}) \cap \mathscr{S}_\alpha} \omega_j ~=~ \int_{(\Sigma_{\tau_2} \times \Sigma_{\tau_2}) \cap \mathscr{S}_\alpha} \omega_j
\end{equation}
is equivalent to
\begin{equation}
\label{eq:ProbconsinSalpha}
\int_{M^{(1)}} \omega_j ~=~ \int_{M^{(2)}} \omega_j
\end{equation}
where $M^{(j)} =  \left\lbrace (t, z_1, t, z_2) \in \partial \mathscr{S}_{\alpha}: \max \{ z_1, z_2 \} = z_j \ \wedge \tau_1 < t < \tau_2 \right\rbrace $ for $j=1,2$. Observe that from $(t, z_1, t, z_2)  \in M^{(j)}$ it follows that $z_j = z_{3-j} + \alpha$. Furthermore, antisymmetry implies: 
\begin{equation}
\omega_j (t, z, t, z + \alpha) ~=~ - \omega_j (t, z + \alpha, t ,z).
 \label{eq:antisymmetricomegaj}
\end{equation}
This can be seen from the fact that on $\mathscr{C}_1$, $\omega_j = (|\psi_{-+}|^2 - |\psi_{+-}|^2) \, d t \wedge d z$ (see the proof of lemma \ref{thm:karltheodor}).\\
Inserting \eqref{eq:antisymmetricomegaj} into eq.\@ \eqref{eq:ProbconsinSalpha} allows us to conclude:
\begin{equation}
\int_{M^{(1)}} \omega_j ~=~ - \int_{M^{(1)}} \omega_j  ~=~ 0.
\end{equation}
As this relation must hold for every $\tau_1, \tau_2$, we must have $\omega_j (p) = 0$.
In components:
\begin{equation}
| \psi_{+-}(p)|^2 - |\psi_{-+}(p)|^2 ~=~ 0 ~~ \Leftrightarrow ~~  \psi_{+-}(p) ~=~ e^{i \varphi(p)} \psi_{-+}(p),
\end{equation}
where $\varphi: \partial \mathscr{S}_\alpha \rightarrow (-\pi, \pi]$ could in principle be a function which is not constant.\\
Because $p$ is an arbitrary boundary point, this equation must hold on the whole of $\partial \mathscr{S}_\alpha$. Moreover, the requirement of Poincar\'e invariance has the consequence that $\varphi(p)$ has to be locally constant (see the example preceding lemma \eqref{thm:MyBCareProbconserving}). The domain $\mathscr{S}_\alpha$ has the two connected components $\mathscr{S}^{\pm}_{\alpha}$ and by antisymmetry one obtains:
\begin{equation}
\left. \varphi \right|_{ \mathscr{S}^{+}_{\alpha}} ~=~ - \left. \varphi \right|_{ \mathscr{S}^{-}_{\alpha}}.
\end{equation} 
Thus, indeed no other boundary conditions than \eqref{eq:GoodBCforSalpha} are permitted. \qed
\end{proof}

\begin{remark}
 \begin{enumerate}
  \item A similar proof for distinguishable particles shows that in this case another possibility appears: the two contributions in eq.\@ \eqref{eq:ProbconsinSalpha} could cancel instead of vanishing individually. However, this cancelling is not physically sensible because it would imply a non-vanishing current from $\mathscr{S}^+_\alpha$ to $\mathscr{S}^+_\alpha$ and vice versa. Provided the Born rule holds, the particles could then swap place instantaneously.
  \item The boundary conditions \eqref{eq:GoodBCforSalpha} do indeed imply Poincar\'e invariance and probability conservation. However, this will not be shown explicitly as they do not lead to the existence of dynamics (see the following lemma).
 \end{enumerate}

\end{remark}

\begin{lemma}
\label{thm:TheoremAlphaNoExistence}

Let $\alpha >0$ and consider the IBVP given by
\begin{equation}
\label{eq:AlphaIBVP}
\left\lbrace \begin{array}{rcl}
i  \gamma_k^{\mu} \partial_{k, \mu} \psi (t_1, z_1, t_2, z_2)  & = & 0 \ \ \mathrm{for} \ \ k= 1, 2,
\\ \psi (0, z_1, 0, z_2) & = & g(z_1, z_2),
\\ \psi_{ + - } &= & e^{i \varphi} \, \psi_{- + } \ \mathrm{on} \ \ \partial \mathscr{S}_\alpha
\end{array} \right.
\end{equation}
on the domain $\mathscr{S}_\alpha$. Here, $\varphi \in (-\pi,\pi]$ and $g: \{ (z_1,z_2) \in \R^2 : |z_1-z_2| > \alpha\} \rightarrow \C^4$ is supposed to be a $C^1$-function.\\
Then, if there exist real numbers $a_1<b_1<a_2<b_2$ with $g_{+-} (a_1, a_2) \neq g_{+-} (b_1, b_2)$ or $g_{-+} (a_1, a_2) \neq g_{-+} (b_1, b_2)$ the IBVP \eqref{eq:AlphaIBVP} does not have any $C^1$-solution.
\end{lemma}

\begin{proof}
Assume that there exist real numbers $a_1<b_1<a_2<b_2$ with $g_{-+} (a_1, a_2) \neq g_{-+} (b_1, b_2)$. The case of $g_{+-}$ is similar and will not be shown explicitly. Suppose that $\psi$ is a solution of \eqref{eq:AlphaIBVP}. We obtain a contradiction by constructing points $(t_1,y_1,t_2,y_2)$ and $(s_1,x_1,s_2,x_2) \in \mathscr{S}_\alpha$ which lie on the same multi-time characteristic with respect to the component $\psi_{+-}$ (see fig.\@ \ref{fig:FigureConstructionProofAlpha}).\\
The construction proceeds as follows:
\begin{enumerate}
\item Choose a point $(t_1, y_1, t_2, y_2)$ on the same multi-time characteristic of $\psi_{+-}$ as $(0, a_1, 0, a_2)$ and on the boundary of $\mathscr{S}_\alpha$, i.e.
\begin{equation} \label{eq:PointsAlphaContradiction1}
\left\lbrace \begin{array}{rcl}
a_1 & = & y_1 - t_1
\\ a_2 & = & y_2 + t_2
\\ (t_1 - t_2)^2 &=& (y_1 - y_2)^2 - \alpha^2.
\end{array} \right.
\end{equation}
This in particular implies: 
\begin{equation}
\label{eq-+durchgY}
\psi_{-+} (t_1, y_1, t_2, y_2) ~=~ g_{-+} (a_1, a_2).
\end{equation}
\item Consider the set of points $(s_1, x_1, s_2, x_2)$ on the same multi-time characteristic as $(0, b_1, 0, b_2)$ and on the boundary of $\mathscr{S}_\alpha$, i.e.
\begin{equation}  \label{eq:PointsAlphaContradiction2}
\left\lbrace \begin{array}{rcl}
b_1 & = & x_1 - s_1
\\ b_2 & = & x_2 + s_2
\\ (s_1 - s_2)^2 &=& (x_1 - x_2)^2 - \alpha^2.
\end{array} \right.
\end{equation}
This means
\begin{equation} 
  \label{eq:-+durchgX}
  \psi_{-+} (s_1, x_1, s_2, x_2) ~=~ g_{-+} (b_1, b_2).
\end{equation}
\item Now select a point $(s_1, x_1, s_2, x_2)$ on the same multi-time charac\-teris\-tic as $(t_1,y_1,t_2,y_2)$ with respect to the component $\psi_{+-}$, i.e.\@
\begin{equation}  \label{eq:PointsAlphaContradiction3}
\left\lbrace \begin{array}{rcl}
x_1 + s_1 &=& y_1 + t_1, 
\\ x_2 -s_2 &=& y_2 - t_2.
\end{array} \right.
\end{equation}
This implies that the value at $(t_1, y_1, t_2, y_2)$ can be obtained in two different ways: firstly by using the boundary condition at that point and secondly by going along the characteristic surface
\footnote{One may wonder how it is possible to have a path connecting the two points which neither leaves the characteristic nor the domain. This is achieved as follows. Concatenate the two linear paths from $(t_1, y_1, t_2, y_2)$ to $(t_1, y_1, s_2, x_2)$ and from $(t_1, y_1, s_2, x_2)$ to $(s_1, x_1, s_2, x_2)$, so first move the right point from $Y_2$ to $X_2$ and afterwards the left from $Y_1$ to $X_1$. One can see from the hyperbolas in figure \ref{fig:FigureConstructionProofAlpha} that this path only leaves $\mathscr{S}_\alpha$ at its endpoints.}
to $(s_1, x_1, s_2, x_2)$ and using the value from there. In formulas: 
\begin{flalign}
\psi_{+-} (t_1, y_1, t_2, y_2) ~&\overset{\mathrm{b.c.}}{=}~ e^{i \varphi} \, \psi_{-+} (t_1, y_1, t_2, y_2) \nonumber
\\  ~&\overset{\eqref{eq-+durchgY}}{=}~ e^{i \varphi} \, g_{-+} (a_1, a_2). \nonumber\\
&~\\
 \psi_{+-} (t_1, y_1, t_2, y_2) ~&\overset{\mathrm{char.}}{=}~ \psi_{+-} (s_1, x_1, s_2, x_2) \nonumber\\
& \overset{\mathrm{b.c.}}{=}~ e^{i \varphi} \, \psi_{-+} (s_1, x_1, s_2, x_2) 
\nonumber\\ &\overset{\eqref{eq:-+durchgX}}{=}~ e^{i \varphi} \, g_{-+} (b_1, b_2).
\end{flalign}
Thus:
\begin{equation}
g_{-+} (b_1, b_2) =  g_{-+} (a_1, a_2),
\label{eq:gequation}
\end{equation}
in contradiction to the assumption.
\end{enumerate}

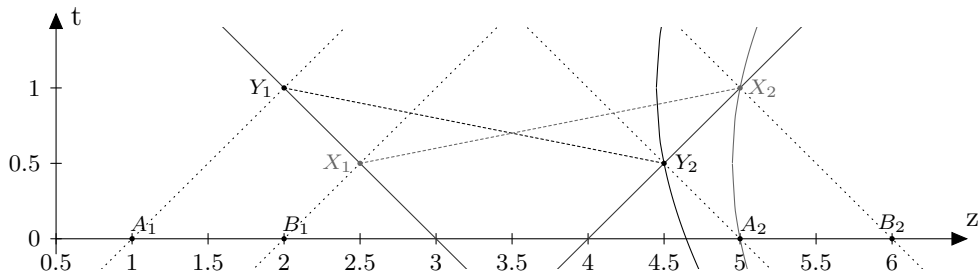
\begin{figure}[h]

\definecolor{tttttt}{rgb}{0.2,0.2,0.2}
\definecolor{wwwwww}{rgb}{0.4,0.4,0.4}
\begin{center}
\begin{tikzpicture}[line cap=round,line join=round,>=triangle 45,x=2.0cm,y=2.0cm]
\draw[->,color=black] (0.5,0) -- (6.5,0);
\foreach \x in {0.5,1,1.5,2,2.5,3,3.5,4,4.5,5,5.5,6}
\draw[shift={(\x,0)},color=black] (0pt,2pt) -- (0pt,-2pt) node[below] {\footnotesize $\x$};
\draw[color=black] (6.42,0.02) node [anchor=south west] { z};

\draw[->,color=black] (0.5,-0.1) -- (0.5,1.5);
\foreach \y in {0, 0.5, 1}
\draw[shift={(0.5, \y)},color=black] (2.0pt,0pt) -- (-2.0pt,-0pt) node[left] {\footnotesize $\y$};

\draw[color=black] (0.53,1.35) node [anchor=south west] { t};

\clip(0.5,-0.2) rectangle (6.5,1.4);
\draw[smooth,samples=100,domain=4.449494000000029:6.5] plot(\x,{1+sqrt(((\x)-2)^2-6)});
\draw[smooth,samples=100,domain=4.449494000000029:6.5] plot(\x,{1-sqrt(((\x)-2)^2-6)});
\draw[color=wwwwww, smooth,samples=100,domain=4.9494919999999984:6.5] plot(\x,{0.5+sqrt(((\x)-2.5)^2-6)});
\draw[color=wwwwww, smooth,samples=100,domain=4.9494919999999984:6.5] plot(\x,{0.5-sqrt(((\x)-2.5)^2-6)});
\draw [dotted,domain=0.5:6.5] plot(\x,{(--1-1*\x)/-1});
\draw [dotted,domain=0.5:6.5] plot(\x,{(--5-1*\x)/1});
\draw [dotted,domain=0.5:6.5] plot(\x,{(--6-1*\x)/1});
\draw [dotted,domain=0.5:6.5] plot(\x,{(--2-1*\x)/-1});
\draw [dash pattern=on 1pt off 1pt] (2,1)-- (4.5,0.5);
\draw [dash pattern=on 1pt off 1pt,color=wwwwww] (2.5,0.5)-- (5,1);
\draw [color=tttttt,domain=0.5:6.5] plot(\x,{(--1.5-0.5*\x)/0.5});
\draw [color=tttttt,domain=0.5:6.5] plot(\x,{(--2-0.5*\x)/-0.5});
\begin{scriptsize}
\fill [color=black] (1,0) circle (1.0pt);
\draw[color=black] (1.08,0.09) node {$A_1$};
\fill [color=black] (2,0) circle (1.0pt);
\draw[color=black] (2.08,0.09) node {$B_1$};
\fill [color=black] (6,0) circle (1.0pt);
\draw[color=black] (6.0,0.10) node {$B_2$};
\fill [color=black] (5,0) circle (1.0pt);
\draw[color=black] (5.09,0.09) node {$A_2$};
\fill [color=black] (2,1) circle (1.0pt);
\draw[color=black] (1.85,1.0) node {$Y_1$};
\fill [color=black] (4.5,0.5) circle (1.0pt);
\draw[color=black] (4.65,0.5) node {$Y_2$};
\fill [color=wwwwww] (2.5,0.5) circle (1.0pt);
\draw[color=wwwwww] (2.35,0.5) node {$X_1$};
\fill [color=wwwwww] (5,1) circle (1.0pt);
\draw[color=wwwwww] (5.15,1.0) node {$X_2$};
\end{scriptsize}

\end{tikzpicture}
\end{center}

\caption{\small Construction in the proof for values $a_1 = 1$, $b_1 = 2$, $a_2 = 5$, $b_2 = 6$ and $\alpha = \sqrt{6}$. The points are $A_j = (a_j, 0)$, $B_j = (b_j, 0)$,  and $Y_j=(y_j, t_j)$, $X_j = (x_j, s_j)$ for $j = 1, 2.$ The black hyperbola consists of points with space-like distance $\alpha$ to $Y_1$ and the grey one of thoses with space-like distance of $\alpha$ to $X_1$. The configurations $(X_1, X_2)$ and $(Y_1, Y_2)$ lie on the same multi-time-characteristic, comprised of the Cartesian product of the two solid black lines.} \label{fig:FigureConstructionProofAlpha}
\end{figure}
\noindent This proves the claim, provided the points we use do exist. Indeed, the combination of the eight equations \eqref{eq:PointsAlphaContradiction1}, \eqref{eq:PointsAlphaContradiction2} and \eqref{eq:PointsAlphaContradiction3} with eight unknowns leads to rather lengthy quadratic equations the general solution of which can be found in appendix \ref{sec:Appendix}. One explicit solution is given in the figure.
\qed
\end{proof}

\begin{remark}
The lemma shows that the most general Lorentz invariant and probability-conserving IBVP \eqref{eq:AlphaIBVP} on $\mathscr{S}_\alpha$ is over-determined. Eq.\@ \eqref{eq:gequation} shows that the only admissible initial data are those for which $g_{-+}$ is constant (and thus also $g_{+-}$). Due to normalization, this constant has to be zero. The two other components are exactly those which are not affected by boundary conditions. Moreover, it becomes clear from the proof that the problem originates from the too high dimension of $\partial \mathscr{S}_\alpha$ which implies (regardless of initial conditions) that certain components of the wave function have to be constant on sets like the initial data surface. One cannot avoid this problem by simply prescribing boundary conditions only on a part of the boundary due to the requirement of Lorentz invariance.
\end{remark}

\section{Discussion}

In this work, we have developed a rigorous, interacting and explicitly solvable relativistic multi-time model for $N$ Dirac particles in $1+1$ dimensions. The main results are (a) the extraction of the class \eqref{eq:probconsbdc} of boundary conditions which are compatible with the requirements of antisymmetry, manifest Lorentz invariance and probability conservation, (b) the proof that uniqueness of solutions of the multi-time equations follows from probability conservation on space-like hypersurfaces, as well as (c) the proof of the existence of dynamics and the explicit formula for solutions. Concerning (a), we believe that this class is the only one compatible with the physical requirements which can be formulated for general $N$ (see the remark at the end of sec.\@ \ref{sec:li}). Furthermore, we showed that the interaction by boundary conditions on sets where the space-time coordinates of two particles coincide can, at equal times, be effectively regarded as given by a spin-dependent $\delta$-potential.\\
Our results show that even the strictest requirements of relativistic invariance, as em\-bodied by the multi-time formalism, can be rigorously satisfied. This makes clear one more time that direct relativistic quantum-mechanical interactions are not generally impossible. \mbox{No-go} theorems about relativistic interactions such as \cite{nogo_potentials,nointeraction} rather rule out only specific mechanisms for interactions. The possibility to explicitly solve the model is instructive for illustrating how the various physical requirements can be met in the multi-time formalism. The model therefore also has a certain pedagogical value. Moreover, it might serve to explicitly test general claims about relativistic quantum mechanics. The fact that the model possesses a conserved tensor current furthermore ensures compatibility with re\-alis\-tic quantum theories such as relativistic GRW models \cite{rel_grw,grwf} and relativistic Bohmian mechanics \cite{hbd_subsystems,hbd,rel_bm}. This is of particular interest because these 
theories have so far only been formulated for the non-interacting case, which is, as our model shows, not due to a problem inherent in these theories but only due to the fact that no rigorous interacting relativistic multi-time theory existed before.\\
Motivated by the question whether a generalization of the model to higher dimensions can be achieved via the introduction of a minimal space-like distance between the particles, we first treated the question of the existence of dynamics for these $\alpha$-spacelike configurations for $d=1$. The result was negative. This leads us to believe that also in higher dimensions a consistent dynamics on the domain of $\alpha$-spacelike configurations does not exist. See, however, \cite{2bd_current_cons} for an analysis of a different multi-time model for two interacting Dirac particles in $1+3$ dimensions.\\
Concerning further generalizations of our model, it should be possible to include non-zero masses. However, one then looses the possibility to explicitly solve the model and consequently a change of the mathematical tools is required (see the discussion of \cite{1d_model}). A further hint that the inclusion of masses should be unproblematic is given by the connection with $\delta$-interactions outlined in sec.\@ \ref{sec:interaction}. Once a self-adjoint extension has been found which implements the $\delta$-interaction, one can always add a bounded symmetric operator such as a mass term to the Hamiltonian. A different generalization would be the case of variable particle numbers. Similarly to Fock space, one could then consider the model on the different $N$-particle sectors and trying to relate them (see \cite{qftmultitime} for a discussion of multi-time quantum field theories). If successful, this would yield a rigorously interacting multi-time QFT model in $1+1$ dimensions which, in addition, might be 
explicitly solvable.

\subsection*{Acknowledgments}
We would like to thank Detlef Dürr for helpful discussions. M.L.\@ gratefully acknowledges financial support by the German National Academic Foundation.

\appendix

\section{Explicit formulas for the points used in the proof of lemma \ref{thm:TheoremAlphaNoExistence}}
\label{sec:Appendix}

In the following we give the solutions of the eight equations \eqref{eq:PointsAlphaContradiction1}, \eqref{eq:PointsAlphaContradiction2} and \eqref{eq:PointsAlphaContradiction3} which are used in the proof of lemma \ref{thm:TheoremAlphaNoExistence}:
\begin{align}
 y_1 ~&=~ a_1+\frac{1}{2} \left(-a_1+b_1+\frac{1}{2} (a_2-2 b_1+b_2)-\frac{1}{2} \xi \right), \nonumber
\\
2t_1 ~&=~  -a_1+b_1+\frac{1}{2} (a_2-2 b_1+b_2)-\frac{1}{2} \xi,
\nonumber
\\
y_2 ~&=~ a_1+\frac{1}{2} \left(-a_1+b_1+\frac{1}{2} (a_2-2 b_1+b_2)+\frac{1}{2} \xi \right),
\nonumber
\\
 t_2 ~&=~ \frac{a_2-b_2+2 \left(\alpha^2-b_1^2+2 b_1 b_2-b_2^2+ \left( b_2 - b_1 \right) \left(\frac{1}{2} (a_2-2 b_1+b_2)-\frac{1}{2} \xi \right) \right)}{ \left(4 b_1-4 b_2+ 2 (a_2-2 b_1+b_2)-2 \xi \right)},
\nonumber
\\
 x_1 ~&=~ b_1+\frac{1}{4} (a_2-2 b_1+b_2)+\frac{1}{4} \xi,
\nonumber
\\
 s_1 ~&=~ \frac{1}{4} (a_2-2 b_1+b_2)+\frac{1}{4} \xi,
\nonumber
\\
 x_2 ~&=~ \frac{b_2-\alpha^2-b_1^2+2 b_1 b_2-b_2^2+ \left( b_2 - b_1 \right) \left(\frac{1}{2} (a_2-2 b_1+b_2)-\frac{1}{2} \xi \right)}{\left(2 b_1-2 b_2+ (a_2-2 b_1+b_2)+ \xi \right)},
\nonumber
\\
 s_2 ~&=~ \frac{\alpha^2-b_1^2+2 b_1 b_2-b_2^2+ \left( b_2 - b_1 \right) \left(\frac{1}{2} (a_2-2 b_1+b_2)-\frac{1}{2} \xi \right)}{\left(2 b_1-2 b_2+ (a_2-2 b_1+b_2)+ \xi \right)},
\end{align}
where
\begin{equation}
\xi ~=~ \sqrt{\frac{(b_2 - a_2)^2 (b_1-a_1) + 4 \alpha^2 (b_2 - a_2)}{b_1-a_1}}.
\end{equation}
The radicand is positive since $a_1<b_1$ and $a_2<b_2$.


\begin{thebibliography}{10}

\bibitem{1d_model}
M.~Lienert.
\newblock {A relativistically interacting exactly solvable multi-time model for
  two mass-less Dirac particles in 1+1 dimensions}, 2014.
\newblock To be published. Preprint: arXiv:1411.2833v2.

\bibitem{dirac_32}
P.~A.~M. Dirac.
\newblock {Relativistic Quantum Mechanics}.
\newblock {\em Proc. R. Soc. Lond. A}, 136:453--464, 1932.

\bibitem{nogo_potentials}
S.~{Petrat} and R.~{Tumulka}.
\newblock {Multi-Time Schrödinger Equations Cannot Contain Interaction
  Potentials}.
\newblock {\em J. Math. Phys.}, 55(032302), 2014.
\newblock arXiv:1308.1065v2.

\bibitem{albeverio}
S.~{Albeverio}, F.~{Gesztesy}, R.~{H{\o}egh-Krohn}, and H.~{Holden}.
\newblock {\em {Solvable Models in Quantum Mechanics}}.
\newblock AMS Chelsea Publishing, 2005.

\bibitem{hbd_subsystems}
D.~{Dürr} and M.~{Lienert}.
\newblock {On the description of subsystems in relativistic hypersurface
  Bohmian mechanics}.
\newblock {\em Proc. R. Soc. A}, 470(2169), 2014.
\newblock arXiv:1403.1464v2.

\bibitem{tomonaga}
S.~Tomonaga.
\newblock {On a Relativistically Invariant Formulation of the Quantum Theory of
  Wave Fields}.
\newblock In J.~Schwinger, editor, {\em {Selected Papers on Quantum
  Electrodynamics}}, pages 156--168. Dover, 1958.

\bibitem{schwinger}
J.~Schwinger.
\newblock {Quantum Electrodynamics. I. A Covariant Formulation}.
\newblock {\em Phys. Rev.}, 74(2162):1439--1461, 1948.

\bibitem{qftmultitime}
S.~{Petrat} and R.~{Tumulka}.
\newblock {Multi-Time Wave Functions for Quantum Field Theory}.
\newblock {\em Ann. Phys.}, 345:17--54, 2014.
\newblock arXiv:1309.0802v2.

\bibitem{hbd}
D.~{Dürr}, S.~{Goldstein}, K.~{Münch-Berndl}, and N.~{Zanghì}.
\newblock {Hypersurface Bohm-Dirac models}.
\newblock {\em Phys. Rev. A}, 60:2729--2736, 1999.
\newblock arXiv:quant-ph/9801070v2.

\bibitem{LukasThesis}
L.~Nickel.
\newblock {Master's thesis. On relativistic interactions in quantum theories}.
  In preparation.
\newblock Mathematical Institute, Ludwig-Maximilians-Universität, München.

\bibitem{kellarstephenson}
B.~H.~J. {McKellar} and G.~J. {Stephenson Jr.}
\newblock {Klein Paradox and the Dirac-Kronig-Penney model}.
\newblock {\em Phys. Rev. A}, 36(6):2566--2569, 1987.

\bibitem{diracpotential}
M.~G. {Calkin}, D.~{Kiang}, and Y.~{Nogami}.
\newblock {Proper treatment of the delta function potential in the one
  dimensional Dirac equation}.
\newblock {\em Am. J. Phys.}, 55(8):737--739, 1987.

\bibitem{nointeraction}
D.~G. Currie, T.~F. Jordan, and E.~C.~G. Sudarshan.
\newblock {Relativistic Invariance and Hamiltonian Theories of Interacting
  Particles}.
\newblock {\em Rev. Mod. Phys.}, 35:350--375, 1963.

\bibitem{rel_grw}
D.~{Bedingham}, D.~{Dürr}, G.C. {Ghirardi}, S.~{Goldstein}, R.~{Tumulka}, and
  N.~{Zanghì}.
\newblock {Matter Density and Relativistic Models of Wave Function Collapse}.
\newblock {\em Journ. of Stat. Phys.}, 154:623--631, 2014.
\newblock arXiv:1111.1425v4.

\bibitem{grwf}
R.~Tumulka.
\newblock {A Relativistic Version of the Ghirardi-Rimini-Weber Model}.
\newblock {\em Journ. of Stat. Phys.}, 125:821--840, 2006.
\newblock arXiv:quant-ph/0406094v2.

\bibitem{rel_bm}
D.~{Dürr}, S.~{Goldstein}, T.~{Norsen}, W.~{Struyve}, and N.~{Zanghì}.
\newblock {Can Bohmian mechanics be made relativistic?}
\newblock {\em Proc. R. Soc. A}, 470(2162), 2014.
\newblock arXiv:1307.1714v2.

\bibitem{2bd_current_cons}
M.~Lienert.
\newblock {On the question of current conservation for the Two-Body Dirac
  equations of constraint theory}, 2015.
\newblock To be published. Preprint: arXiv:1501.07027v1.

\end{thebibliography}

\end{document}